\documentclass[acmsmall]{acmart}
\usepackage[ruled,linesnumbered]{algorithm2e}
\usepackage{xcolor}
\usepackage{algpseudocode}
\usepackage{multirow}
\usepackage{diagbox}
\usepackage{makecell}
\usepackage{amsmath}
\usepackage{booktabs}
\usepackage{threeparttable}
\usepackage{graphics}
\usepackage{epsfig}
\usepackage{setspace}
\usepackage{xspace}
\usepackage{extarrows}
\usepackage{bm}
\usepackage{tabu}
\usepackage{subfig}
\usepackage{color}
\usepackage{xcolor}
\definecolor{keywordcolor}{rgb}{0.8,0.1,0.5}
\usepackage{listings}
\usepackage{tikz}
\newcommand*{\circled}[1]{\lower.7ex\hbox{\tikz\draw (0pt, 0pt)%
    circle (.4em) node {\makebox[1em][c]{\small #1}};}}
\usepackage{enumitem}
\usepackage{eucal} 
\newcommand{\blue}[1]{{\color[rgb]{0,0,1}#1}}

\newcommand{\red}[1]{{\color[rgb]{0.5,0,0}#1}}
\newcommand{\mainname}{\textit{MetHyl}\xspace}
\newcommand{\plusname}{$\textit{MetHyl}^{+}$\xspace}

\newcommand{\cat}{+ \!\!\! + \,}

\newcommand{\F}{\mathsf{F}}
\newcommand{\hylo}[1]{[ \! [  #1 ] \! ] } 

\makeatletter
\newcommand*\bigcdot{\mathpalette\bigcdot@{.5}}
\newcommand*\bigcdot@[2]{\mathbin{\vcenter{\hbox{\scalebox{#2}{$\m@th#1\bullet$}}}}}
\makeatother

\makeatletter
\newcommand{\removelatexerror}{\let\@latex@error\@gobble}
\makeatother

\acmJournal{PACMPL}
\acmVolume{1}
\acmNumber{POPL} 
\acmArticle{1}
\acmYear{2020}
\acmMonth{1}
\acmDOI{} 
\startPage{1}

\setcopyright{none}

\usepackage{booktabs}   

\begin{document}
 

  \title[Automating Thinning Theorem: Synthesizing Efficient Dynamic Programming Algorithms]{Automating Thinning Theorem: Synthesizing Efficient Dynamic Programming Algorithms}

  \author{Ruyi Ji}
  \affiliation{
    \streetaddress{Key Lab of High Confidence Software Technologies, Ministry of Education Department of Computer Science and Technology, School of Computer Science}
    \institution{Peking University}
    \city{Beijing}
    \country{China}            
  }
  \email{jiruyi910387714@pku.edu.cn}          
  
  \author{Tianran Zhu}
  \affiliation{
    \streetaddress{Key Lab of High Confidence Software Technologies, Ministry of Education Department of Computer Science and Technology, School of Computer Science}
    \institution{Peking University}
    \city{Beijing}
    \country{China}            
  }
  \email{moiezen@pku.edu.cn}   
  
  \author{Yingfei Xiong}
  \authornote{Corresponding author}
  \affiliation{
    \streetaddress{Key Lab of High Confidence Software Technologies, Ministry of Education Department of Computer Science and Technology, School of Computer Science}
    \institution{Peking University}
    \city{Beijing}
    \country{China}            
  }
  \email{xiongyf@pku.edu.cn}          
  
  \author{Zhenjiang Hu}
  \affiliation{
    \streetaddress{Key Lab of High Confidence Software Technologies, Ministry of Education Department of Computer Science and Technology, School of Computer Science}
    \institution{Peking University}
    \city{Beijing}
    \country{China}            
  }
  \email{huzj@pku.edu.cn}          

\begin{abstract} Dynamic programming is an important optimization technique, but designing efficient dynamic programming algorithms can be difficult for even professional programmers. Thinning, a technique developed for systematically deriving efficient dynamic programming algorithms, has received much attention in studies because of its effectiveness for a large class of problems. Despite the success of thinning in theory, its practical usage is still limited because (1) applying thinning requires mathematical and algorithmic background, and (2) applying thinning solely may not be enough to generate algorithms as efficient as proposed by human experts.
 
In this paper, we propose two approaches, \mainname and \plusname, to resolve both problems. First, \mainname automates the application of thinning via program synthesis, and thus eliminates the burden to the user for applying thinning. Second, \plusname integrates three rules into \mainname that optimizes three important factors on the time complexity of dynamic programming algorithms that are ignored by thinning, and thus make it able to automatically generate expert-level dynamic programming algorithms on many tasks.

We evaluate our approaches on $37$ tasks related to $16$ optimization problems collected from \textit{Introduction to Algorithm}, a popular textbook for algorithm courses. The results show that \plusname achieves exponential speed-ups on $97.3\%$ tasks with an average time cost of less than one minute. Moreover, \plusname generates algorithms that are as efficient as the reference programs provided by human experts on $70.3\%$ tasks.
\end{abstract}


\maketitle
 
\section{Introduction} \label{section:introduction}
\textit{Combinatorial Optimization} is a topic on finding an optimal solution from a finite set of valid solutions~\cite{schrijver2003combinatorial}. Combinatorial optimization problems (COPs), such as the knapsack problem and the traveling salesman problem, widely exist in various domains. Solving a COP is usually difficult as the number of valid solutions can be extremely large. 

\textit{Dynamic programming} is an important technique for solving COPs. A dynamic programming algorithm can be implemented in the top-down approach or the bottom-up approach, where the top-down approach is also known as \emph{memoization}. Given a recursive function, memoization can be easily implemented by caching the results of existing calls. Though obtaining an arbitrary memoization algorithm is trivial, different memoization algorithms could have huge performance differences. Designing an \emph{efficient} memoization algorithm for a specific problem is difficult and takes algorithmic efforts. 

Motivated by the importance and difficulty of designing efficient memoization, many approaches have been proposed for systematically transforming a plain program into an efficient memoization algorithm. In this paper, we consider one important approach among them, namely \textit{thinning}~\cite{DBLP:books/daglib/0096998}. Thinning transforms a plain program specified by a recursive generator, which generates all valid solutions, and an objective function, which evaluates the objective value for each solution, into a more efficient program that does not consider most solutions.
By previous studies~\cite{DBLP:conf/plilp/Moor95, DBLP:conf/flops/MorihataKO14,DBLP:conf/pepm/Mu08,DBLP:books/daglib/0096998, DBLP:journals/jfp/Bird01a, sasano2000make, DBLP:journals/ngc/Morihata11}, thinning can derive efficient memoization for a large class of COPs.

However, despite the success of thinning in theory, its help to average programmers is still limited because of two shortages. First, grasping the usage of thinning requires mathematical and algorithmic background. On the one hand, the formal definition of thinning is highly abstracted and involves concepts in the category theory. On the other hand, thinning requires the user to provide a proper preorder on solutions, and in many cases, finding such a preorder is non-trivial and relies on algorithmic intuitions. Therefore, learning and using thinning are both difficult for most programmers.

Second, though thinning can generate efficient memoization algorithms for many COPs, applying this approach solely is usually not enough to achieve an algorithm that is as efficient as the one proposed by human experts. In our evaluation, the time complexity of the memoization derived by thinning is asymptotically larger than the reference solutions on $34/37(91.9\%)$ tasks in our dataset.  

In this paper, we make two contributions to resolve these shortages respectively. For the first shortage, to remove the burden from the user, we show that the application of thinning can be fully automated via program synthesis. \textit{The first contribution of this paper is a fully automated approach for thinning, namely \mainname}. \mainname treats the application of thinning as a program synthesis task for the preorder and follows the framework of programming-by-example~\cite{DBLP:conf/ijcai/ShawWG75}. Given a recursive generator, an objective function, and several concrete instances of the COP, \mainname extracts a set of examples for the preorder according to the theory of thinning, where each example specifies that the effectiveness of one solution is not dominated by another. \mainname synthesizes a valid preorder from the examples via a novel synthesis algorithm, and then automatically generates an efficient memoization algorithm by thinning with the synthesized preorder. To use \mainname, the user needs neither to find out a proper preorder him/herself nor to learn anything about thinning. In this way, the difficulty of using thinning in practice is greatly reduced.

We implement \mainname and evaluate it on $37$ tasks collected from \textit{Introduction to Algorithm}~\cite{cormen2009introduction}, a popular textbook for algorithm courses. The results show that (1) \mainname successfully synthesizes a preorder for thinning on $36/37 (97.3\%)$ tasks with an average time cost of $4.21$ seconds, and (2) the program generated by \mainname achieves exponential speed-ups against the plain program on $31/37 (83.8\%)$ tasks. As mentioned before, we also compare the results of \mainname with the reference solutions to COPs in \textit{Introduction to Algorithm} provided by \citet{cormen2009introduction} and \citet{algsolutions}. The results demonstrate the gap between thinning and human experts: \mainname achieves the same time complexity as the reference solution on only $3/37 (8.1\%)$ tasks.

For the second shortage, to further improve the memoization generated by \mainname, we analyze the time complexity of memoizing a recursive generator and show that it is determined by four different factors: (1) the number of solutions returned by the generator, (2) the time cost of constructing solutions from the recursive results, (3) the number of memoized search states, and (4) the time cost of constructing recursive search states from the current one. The main shortage of thinning is that it focuses only on the first factor, while a human expert can make a comprehensive optimization on all four factors. Therefore, we also consider the other three factors. \textit{The second contribution of this paper is three other rules and their automation for the remaining three factors}.\begin{itemize}
    \item For the second factor, to reduce the time cost of constructing solutions, our rule replaces the solutions in the plain program, which usually involves inductive data structures, into a tuple comprising a small number of scalar values while keeping the behavior unchanged. In this way, the time cost is greatly reduced, as manipulating a tuple is usually much faster than an inductive data structure.
    \item For the third factor, to reduce the number of memoized search states, our rule requires a proper equivalence relation over search states and optimizes by skipping those search states of which an equivalent search state has been memoized before.
    \item The rule for the fourth factor is similar to the rule for the second factor. It reduces the time cost of constructing search states by replacing them with tuples.
\end{itemize}
Similar to thinning, the automation of these rules also involves program synthesis and follows the framework of programming-by-example. In theory, we show the effectiveness of these three rules: Under certain assumptions, thinning and the three rules together are guaranteed to reduce the time complexity of the input program to pseudo-polynomial.

We integrate these rules into \mainname as \plusname and evaluate it on our dataset. The results demonstrate that the improvement brought by the three rules is significant. First, \plusname achieves exponential speed-ups on $36/37 (97.3\%)$ tasks with an average time cost of $59.2$ seconds. Second, on $26/37 (70.3\%)$ tasks, \plusname achieves the same time complexities as the reference solutions.
\section{Overview} \label{section:overview}
In this section, we introduce the main ideas of thinning, \mainname, and \plusname using a classical COP namely \textit{0/1 knapsack}~\cite{mathews1896partition}. 

\smallskip
\parbox{0.9\columnwidth}{\begin{center}\it
    Given a set of items, each with a weight $w_i$ and a value $v_i$, put a subset of them in a knapsack of capacity $W$ to get the maximum total value in the knapsack.
\end{center}}
\smallskip

\noindent For example, $xs = [(3, 3), (2, 2), (1, 2)], W = 4$ describes an instance of 0/1 knapsack, where three items are available, their weights are $3, 2, 1$ respectively, their values are $3, 2, 2$ respectively, and the capacity of the knapsack is $4$. At this time, the optimal solution is to put the first and the third items, i.e., $[(3, 3), (1,2)]$, where the value sum is $5$. 

In this section, we assume that the number of the items (i.e., $|xs|$) and the capacity (i.e., $W$) are on the same magnitude, denoted by $O(n)$. At this time, there is a standard dynamic programming algorithm for 0/1 knapsack, which runs in $O(n^2)$ time.

\subsection{Problem Specification and Memoization}
To formally describe a COP, we need to specify (1) the set of valid solutions, and (2) the objective value of each solution. In a natural specification of 0/1 knapsack, the set of valid solutions is all subsets of items whose total weight is within the capacity, and the objective value is the total value of items in the subset. 

In this paper, we assume these two parts are specified by two programs $g$ and $o$ respectively. 
\begin{itemize}
    \item The generator $g$ takes the parameters of the problem (in this case, the list of items) as the input and generates all valid solutions for the problem. 
    \item The scorer $o$ is an objective function that maps each solution to its objective value.
\end{itemize}

The code in Figure \ref{figure:01spec} shows one such specification $(g, o)$ for 0/1 knapsack. The parameter $g$ of function $g$ is itself to enable recursion later with a fixed-point combinator and the parameter $xs$ is the list of items. Functions {\it sumw} and {\it sumv} are used to calculate the sum of weights and values for a list of items, respectively. Given a fixed-point combinator \textit{fix}, $\textit{fix}\ g$ exhaustively returns all sublists of items whose total weight is within the capacity, and $o$ calculates the total value of items. For simplicity, we shall directly use $g$ to refer to the recursive version. Note that such a specification may not be unique. Two other specifications for 0/1 knapsack can be found in Section \ref{subsection:program}.

\begin{figure*}
\small
\begin{minipage}[c]{.49\linewidth}{
    \begin{align*}
        &g =  \lambda g. \lambda xs. \text{ if }|xs| = 0\text{ then } [[]] \\
        & \qquad \text{else } (g\ (\textit{tail}\ xs)) \cat \\
        &\qquad \quad \big[(head\ xs)::p\ \big|\ p \in g\ (\textit{tail}\ xs), \\
        &\qquad \qquad \quad (\textit{sumw}\ p) + (head\ xs).1 \leq W\big] \\
        &o =  \lambda p. (\textit{sumv}\ p)
    \end{align*}
    \vspace{-2em} 
    \caption{One specification $(g, o)$ for 0/1 knapsack, where $W$ is a global variable representing the capacity.}\label{figure:01spec}
}\end{minipage}\ \ \ 
\begin{minipage}[c]{.45\linewidth}{
    \begin{alignat*}{2}
        &m &=\,& \lambda xs. xs \\
        &mem &\,=\,& \lambda mem. \lambda xs.  \text{ if } \textit{buffer}[m\ xs] = \bot \text{ then } \\
        & && \quad\textit{buffer}[m\ xs] \leftarrow g\ mem\ xs;  \\
        &&& \textit{buffer}[m\ xs] \\
        &prog &\,=\,& \lambda xs. \textit{argmax}\ o\ ((\textit{fix}\ mem)\ xs)
    \end{alignat*}
    \vspace{-1.3em} 
    \caption{A template for transforming programs $(g, o)$ to a memoization algorithm.}  \label{figure:memoization}
}\end{minipage}
\end{figure*}

     
We can easily construct a memoization algorithm for 0/1 knapsack from specification $(g, o)$ via a template shown in Figure \ref{figure:memoization}. In this template, $m$ is a function that returns a key for an invocation of $g$. Two invocations have the same key only if their outputs are the same. Currently, $m$ is the identity function to trivially ensure this property. Function $mem$ implements the memoization algorithm and \textit{buffer} is a global map that stores the result for each key. Finally, $prog$ returns the optimal solution from all solutions. Here $\textit{argmax}\ o\ ps$ chooses the optimal solution in a list of solutions $ps$ based on the objective function $o$, and $\textit{fix}$ is a fixed-point combinator.

Though the template has effectively reused repeated invocations to $g$, the time complexity of the generated algorithm is still exponential to the number of items. Compared to the standard $O(n^2)$-time algorithm for 0/1 knapsack, such a trivial memoization algorithm is unsatisfactory.

\subsection{Thinning and Its Shortages} 
Before discussing the derivation of efficient memoization algorithms, we first introduce two notations for the convenience of presentation. To distinguish the input and the output of the outermost invocation to $g$ and recursive invocations, we call the input of an arbitrary invocation to $g$ a \textit{search state} (or \textit{state}), as an invocation represents a step in a depth-first search, and call the solution generated by a recursive invocation to $g$ a \textit{partial solution}, as it is not a full solution yet.

\noindent \textbf{The Main Idea of Thinning}. The main reason for the ineffectiveness of the trivial memoization algorithm is that the number of partial solutions returned by the generator $g$ can be exponential to the number of items, and memoization keeps this factor unchanged. Therefore, to generate an efficient memoization algorithm, it is important to reduce the number of partial solutions. 

\textit{Thinning}, proposed by \citet{DBLP:books/daglib/0096998}, is such an approach. As shown in the following program $g'$, to reduce the number of solutions returned by $g$, thinning inserts a special function $\textit{thin}$ at the return point of $g$, which prunes off non-optimal solutions from those generated by $g$.
\begin{equation}
g' = \lambda g'. \lambda xs. \textit{thin}[R]\ (g\ g'\ xs) \label{equation:thinning}
\end{equation}

Function $\textit{thin}$ is parameterized by a preorder $R$ over the space of solutions. Intuitively, thinning requires preorder $R$ to specify the domination between partial solutions: For any two partial solutions $(p_1, p_2)$, $p_1$ is worse than $p_2$ in the sense of $R$ (written as $p_1Rp_2$) only if partial solution $p_1$ can never lead to the global optimal solution if $p_2$ exists. Given preorder $R$ and a set of partial solutions $P$, $\textit{thin}[R]\ P$ returns one smallest subset of $P$ such that all partial solutions in $P$ are dominated by partial solutions in this subset. According to the requirement to $R$, $\textit{thin}[R]$ removes only those non-optimal partial solutions generated by $g$ and does not affect the final result. 

Let us take 0/1 knapsack and its specification $(g,o)$ introduced in Figure \ref{figure:01spec} as an example. In this case, a partial solution is a sub-list of some suffix of the full item list, and it will be completed to a full solution by inserting some (possibly none) items to its front. Therefore, partial solution $p_1$ is dominated by $p_2$ if $p_1$ not only consumes more capacity but also gains smaller value. At this time, any full solution leaded by $p_1$ can be improved by replacing items in $p_1$ with $p_2$. Such a relation can be described by the following preorder $R$.
\begin{equation*}
p_1Rp_2 \iff (\textit{sumw}\ p_1 \geq \textit{sumw}\ p_2) \wedge (\textit{sumv}\ p_1 \leq \textit{sumv}\ p_2)
\end{equation*}

There are two noticeable properties of $\textit{thin}[R]$ that makes $g'$ (Equation \ref{equation:thinning}) efficient. 
\begin{itemize}
    \item First, the number of partial solutions returned by $\textit{thin}[R]$, which is equal to the number of solutions returned by $g'$, is bounded by the maximum outputs of $\textit{sumw}$, which is at most $W$. 
    \item Second, there is an efficient implementation of $\textit{thin}[R]$ that invokes preorder $R$ only linear times, where the time complexity of $R$ is linear to the size of partial solutions. 
\end{itemize}
Therefore, the time complexity of memoizing $(g', o)$ is only $O(n^3)$, which is exponentially faster than directly memoizing the specification $(g, o)$. For simplicity, we do not go deep into $\textit{thin}[R]$ in this section. A detailed discussion on $\textit{thin}$ can be find in Section \ref{section:thinning}, which shows that the efficiency of $\textit{thin}[R]$ is related to the range of \textit{keys} (in this case, $\textit{sumw}$ and $\textit{sumv}$) involved in $R$.

\noindent \textbf{The Shortages of Thinning}. So far, we have successfully obtained a polynomial-time memoization algorithm for 0/1 knapsack by applying thinning to a plain specification. However, our previous discussion also exposes two crucial shortages of thinning.

First, thinning requires the user to provide a preorder that specifies the domination between partial solutions. However, finding such a preorder is a non-trivial task and may rely on algorithmic intuitions. In our example, to find a proper preorder for 0/1 knapsack, the user needs to recognize and involve the comparison between the consumed capacity. Actually, the difficulty of finding such a comparison is already close to directly proposing the standard dynamic programming algorithm for 0/1 knapsack, which takes the consumed capacity as the state.

Second, applying thinning solely is not enough to achieve an algorithm that is as efficient as the one proposed by human experts. In our example, there is still a gap between the result of thinning, which runs in $O(n^3)$ time, and the standard $O(n^2)$-time dynamic programming algorithm.

\subsection{\mainname: Automating Thinning via Programming-by-Example} \label{section:moti-main}
For the first shortage, to remove the burden from the user, one natural way is to automate the application of thinning. If a proper preorder can be found automatically, thinning can be treated as a black-box, and thus applying it will not consume any user's effort. In this paper, we propose \mainname, which efficiently synthesizes preorders for thinning via programming-by-example.

\smallskip
\noindent \textbf{Specification for the Preorder}. The thinning theorem proposed by \citet{DBLP:books/daglib/0096998} provides a formal characterization for the correctness of thinning. \mainname takes this theory as the specification and synthesizes the preorder from it.

The thinning theorem requires the generator to be specified as a (relational) hylomorphism, a common template for recursions in functional programming. In a nutshell, a hylomorphism specification for the generator comprises two separate functions $\psi$ and $\phi$. 

First, $\psi$ generate a set of \textit{transitions} for a given state, where each transition includes several (possibly none) sub-states for recursion and information used to construct solutions. For example, $\psi$ corresponds to the generator $g$ in Figure \ref{figure:01spec} may return the following three transitions.
\begin{enumerate}
\item An empty transition when the item list $xs$ is empty.
\item A direct recursion to state $(\textit{tail}\ xs)$, representing that the first item is skipped. 
\item A recursion to $(\textit{tail}\ xs)$ with information $(\textit{head}\ xs)$, representing that the first item is chosen.
\end{enumerate}

Second, $\phi$ constructs a set of solutions for each transition and each partial solution of sub-states. For example, the following describes the $\phi$ corresponding to the generator $g$ in Figure \ref{figure:01spec}.
\begin{enumerate}
\item For an empty transition, $\phi$ returns an empty list, representing an empty knapsack.
\item For a direct transition, given a partial solution of the sub-state, $\phi$ directly returns the partial solution, representing that no item is added to the knapsack. 
\item For the last transition, given an item and partial solution of the sub-state, $\phi$ adds the item to the partial solution when the capacity is enough, and otherwise returns noting.
\end{enumerate}

Our approaches inherit the requirement on a hylomorphism-style generator from the thinning theorem. We design a language for implementing such programs in Section \ref{subsection:program} and discuss the effect of this requirement on our approaches in Section \ref{section:discussion}. For simplicity, we leave the formal definition of hylomorphism to Section \ref{section:hylomorphism} and still use functions to demo our approaches in this section.

The main advantage of introducing hylomorphism is that how a partial solution is constructed from the partial solutions of sub-states is explicitly specified via function $\phi$. Concretely, relation $p \twoheadrightarrow_{s} p'$, denoting that a partial solution $p'$ of state $s$ can be constructed from another partial solution $p$, can be extracted from the invocations of $\phi$. $p \twoheadrightarrow_{s} p'$ holds if and only if there is an invocation of $\phi$ in state $s$ that takes $p$ as the input and generates $p'$. 

The following lists all instances of this relation for 0/1 knapsack specified in Figure \ref{figure:01spec} when the state $xs$ is $[(3, 3), (2, 2), (1, 2)]$ and the capacity $W$ is $4$.
\smallskip
\begin{gather}
[] \twoheadrightarrow_{\textit{xs}} [] \quad [] \twoheadrightarrow_{\textit{xs}} [(3, 3)] \quad [(1, 2)]\twoheadrightarrow_{\textit{xs}} [(1, 2)] \quad [(1, 2)] \twoheadrightarrow_{\textit{xs}} [(3, 3), (1, 2)] \nonumber \\
[(2, 2)] \twoheadrightarrow_{\textit{xs}} [(2, 2)] \quad [(2,2), (1, 2)] \twoheadrightarrow_{\textit{xs}} [(2, 2), (1, 2)] \label{example:arrow}
\end{gather}

In our example, to ensure the correctness, i.e., the equivalence between $(g, o)$ (Figure \ref{figure:01spec}) and $(g', o)$ (Equation \ref{equation:thinning}), the thinning theorem requires preorder $R$ to satisfy the following conditions.
\begin{enumerate}
\item The dominance specified by $R$ is monotonic during the recursion. If partial solution $p_1$ is worse than $p_2$ in the sense of $R$, all partial solutions generated by $p_1$ must also be worse than those generated by $p_2$. Concretely, the thinning theorem requires the following condition to hold for any two consecutive states $xs$ and $(\textit{tail}\ xs)$, any two partial solutions $p_1, p_2$ of state $(\textit{tail}\ xs)$, and any partial solution $p_1'$ of state $xs$ such that $p_1 \twoheadrightarrow_{\textit{xs}} p_1'$. 
\begin{equation}p_1Rp_2 \rightarrow \exists p_2' \in (g\ xs), \big(p_2 \twoheadrightarrow_{\textit{xs}} p_2' \wedge p_1'Rp_2'\big)\label{formula:monotonic}\end{equation}
\item Preorder $R$ implies the order of the objective value. If partial solution $p_1$ is worse than $p_2$ in the sense of $R$, the objective value of $p_1$ must be no larger than $p_2$, i.e., $p_1Rp_2 \rightarrow (o\ p_1 \leq o\ p_2)$.
\end{enumerate}

Intuitively, the first condition ensures that $(g'\ xs)$ is always equivalent to $(\textit{thin}[R]\ (g\ xs))$, and the second condition ensures that the optimal solution with the largest objective value is always reserved by thinning. Therefore, they together imply the correctness of thinning.

\smallskip
\noindent \textbf{Extracting Examples for the Preorder}. The thinning theorem provides a specification, and the remaining task for automating thinning is to find a preorder $R$ satisfying both conditions. 

However, directly synthesizing from these conditions is challenging due to the complexity of Formula \ref{formula:monotonic}, which involves both universal quantifiers $\forall$ and $\exists$, and a possibly complex relation $\twoheadrightarrow$ defined on the semantics of the input function $\varphi$. To our knowledge, there is no efficient synthesizer that can handle such a complex specification.

\mainname uses the framework of \textit{programming-by-example (PBE)}~\cite{DBLP:conf/ijcai/ShawWG75} to resolve this challenge. Given a logic specification, a typical PBE solver first substitutes concrete values into the formula, extracts constraints on \textit{concrete} invocation of the target program (denoted as examples), and then synthesizes from the examples. In this way, the core synthesizer does not have to handle complex logic specifications, and thus the difficulty of synthesis is greatly reduced. 
\smallskip

Now, we outline how \mainname extracts simple examples for $R$ from the two conditions. 

First, \mainname ensures the second condition by limiting the form of $R$ to be $\lambda p_1. \lambda p_2. (o\ p_1 \leq o\ p_2) \wedge p_1 \red{?R} p_2$, where $\red{?R}$ is a preorder to be synthesized, and thus considers only Formula \ref{formula:monotonic}.

Then, \mainname considering the following formula that is equivalent to Formula \ref{formula:monotonic} for any $p_1, p_2$. 
\begin{equation}\exists p_1' \in (g\ xs), \bigg(p_1 \twoheadrightarrow_{xs} p_1' \wedge  \forall p_2' \in (g\ xs), \big(\neg p_2 \twoheadrightarrow_{\textit{xs}} p_2' \vee \neg p_1'Rp_2'\big)\bigg) \rightarrow \neg p_1Rp_2 \label{formula:eq-mono}\end{equation}

The conclusion of this formula, $\neg p_1Rp_2$, is extremely simple. If the premise can be transformed to be irrelevant to the unknown preorder $R$, we can extract example $(p_1, p_2)$ specifying $\neg p_1Rp_2$ by constantly substituting $p_1, p_2$ with concrete partial solutions until the premise is satisfied. Compared to Formula \ref{formula:monotonic}, the constraint provided by this example involves neither universal quantifiers nor $\twoheadrightarrow$, and thus makes it possible to design an efficient synthesizer for $R$. 

The transformation of Formula \ref{formula:eq-mono} is motivated by the fixed form of $R$ for ensuring the first condition, which implies that a part of $R$ is known while synthesis. Therefore, we substitute $R$ in the conclusion of Formula \ref{formula:eq-mono} with $\lambda p_1. \lambda p_2. p_1R'p_2 \wedge p_1\red{?R}p_2$, where $R'$ and $\red{?R}$ represent the known preorder and unknown preorder in $R$ respectively, and obtain the following equivalent formula.
\begin{align*}\bigg(p_1 R' p_2 \wedge \exists p_1' \in (g\ xs), \bigg(p_1 \twoheadrightarrow_{\textit{xs}} p_1' \wedge  \forall p_2' \in (g\ xs), \big(\neg p_2 \twoheadrightarrow_{\textit{xs}} p_2' \vee \neg p_1'Rp_2'\big)\bigg)\bigg) \rightarrow \neg p_1\red{?R}p_2
\end{align*}

Because $\neg p_1' R' p_2$ implies $\neg p_1' R p_2'$, a \textit{weaker} formula whose premise is irrelevant to $\red{?R}$ can be obtained by replacing $\neg p_1' R p_2'$ in the premise with the unknown comparison $\neg p_1' R' p_2$.
\begin{equation}\bigg(p_1 R' p_2 \wedge \exists p_1' \in (g\ xs), \bigg(p_1 \twoheadrightarrow_{\textit{xs}} p_1' \wedge  \forall p_2' \in (g\ xs), \big(\neg p_2 \twoheadrightarrow_{\textit{xs}} p_2' \vee \neg p_1'R'p_2'\big)\bigg)\bigg) \rightarrow \neg p_1\red{?R}p_2 \label{formula:example}
\end{equation}

We use an example to show how examples are extracted from Formula \ref{formula:example}. In 0/1 knapsack specified by Figure \ref{figure:01spec}, suppose state $xs$ is $[(3, 3), (2, 2), (1, 2)]$, capacity $W$ is $4$, and the known part $R'$ is $\lambda p_1. \lambda p_2. \textit{sumv}\ p_1 \leq \textit{sumv}\ p_2$. At this time, the domain of $p_1, p_2$ is $\{[], [(1, 2)], [(2, 2)], [(2, 2), (1, 2)]\}$ and the relation $\twoheadrightarrow_{\textit{xs}}$ is described in Example \ref{example:arrow}. 
\begin{itemize}
\item When $p_1$ and $p_2$ are taken as $[(1, 2)]$ and $[(2, 2)]$, the premise of Formula \ref{formula:example} is true because (1) $\textit{sumv}\ p_1 \leq \textit{sumv}\ p_2$, (2) $p_1 \twoheadrightarrow_{\textit{xs}} p_1' = [(3, 3), (1, 2)]$, and $\textit{sumv}\ p_1' > \textit{sumv}\ [(2, 2)]$, which is only choice of $p_2'$ satisfying $p_2 \twoheadrightarrow_{\textit{xs}} p_2'$. Therefore, example $\neg [(1, 2)] \red{?R} [(2, 2)]$ is obtained.
\item Similarly, another example can be obtained by taking $p_1$ and $p_2$ as $[(1, 2)]$ and $[(2, 2), (1, 2)]$.
\end{itemize}
Given these two examples, $\red{?R} = \lambda p_1. \lambda p_2. \textit{sumw}\ p_1 \geq \textit{sumw}\ p_2$ is a valid solution satisfying both examples, which leads to the intended preorder for 0/1 knapsack. 

\smallskip
Though Formula \ref{formula:example} is not equivalent to the original specification, i.e., a preorder satisfying Formula \ref{formula:example} may not be valid for Formula \ref{formula:monotonic}, the following fact makes it useful for synthesizing $R$. 
\begin{itemize}
\item Given preorder $R'$, if no examples for $\red{?R}$ can be extracted from Formula \ref{formula:example}, or in other word, the premise of Formula \ref{formula:example} is constantly true, $R'$ must be a valid for Formula \ref{formula:monotonic}. 
\end{itemize}

This fact suggests an iterative framework for synthesizing $R$. Starting from $R' = \lambda p_1. \lambda p_2. o\ p_1 \leq o\ p_2$, a preorder satisfying Formula \ref{formula:monotonic} can be synthesized in three steps. 
\begin{enumerate}
\item Extract examples for $\red{?R}$ from Formula \ref{formula:example} and a set of concrete instances of the COP task.
\item If no examples are obtained, return $R = R'$ as the synthesis result.
\item Synthesize a valid preorder $\red{?R}$ from the examples, update $R'$ with $\lambda p_1. \lambda p_2. p_1R'p_2 \wedge p_1 \red{?R} p_2$, and then go back to Step 1.
\end{enumerate}

\smallskip
\noindent \textbf{Synthesizing a Preorder from Example.} The remaining task for \mainname is to synthesize a preorder $\red{?R}$ satisfying a set of negative examples $(a_i, b_i)$ such that $\neg (a_i \red{?R} b_i)$ holds. Besides, to generate an efficient memoization algorithm, another goal of the synthesis is to minimize the time complexity of the memoization algorithm generated by thinning.

To enable an efficient synthesis algorithm, \mainname synthesizes $\red{?R}$ in the following form, which is a conjunction of comparisons related to several key functions. 
$$
a\red{?R}b \iff \wedge_{i} (\red{?key_i}\ a)\ \red{?op_i}\ (\red{?key_i}\ b)
$$
where $\red{?key_i}$ is a function mapping a partial solution to an integer and $\red{?op_i}$ is an operator in $\{\leq, =, \geq \}$. In our implementation, $\red{?key_i}$ is from a grammar including common arithmetic operators and operators for lists and binary trees. More details on this grammar can be found in Section \ref{section:implementation}. 

This form of preorder has two main advantages. First, a preorder in this form can be naturally decomposed to several comparisons $(\red{?op_i}, \red{?key_i})$, where the scale of each comparison is smaller than the preorder. Moreover, in the terms of satisfying the given negative examples, these comparisons are independent of each other: preorder $\red{?R}$ satisfies a negative example if and only if some comparison in $\red{?R}$ is violated on the example. This property makes it possible to synthesize each comparison separately and thus greatly reduce the scale of the synthesis task. 

Second, we prove that the efficiency of the algorithm generated by thinning with a preorder in this form can be estimated by the production of the ranges of all $\red{?key_i}$. Such estimation is monotonic while including more comparisons to the preorder, and thus is easy to optimize in a search procedure. More details on this estimation can be found in Section \ref{section:thinning}.

\smallskip
Motivated by both properties, \mainname regards a preorder as a list of comparisons and synthesizes the list incrementally. Starting from an empty list, in each turn, \mainname finds a comparison that is violated on a large enough subset of unsatisfied examples and inserts it to the preorder. The iteration proceeds until all examples are satisfied.

To efficiently find an effective preorder, \mainname makes two changes on this basic iteration.
\begin{itemize}
\item As mentioned before, in Section \ref{section:thinning}, we prove that the effectiveness of a preorder for thinning can be estimated by the ranges of the involved key functions. To find an effective preorder, \mainname backtracks on the iteration and uses branch-and-bound, a standard search technique, to optimize the objective function provided by our estimation.
\item To restrain the search space, \mainname uses an outermost iteration on two parameters: (1) a size limit $n_s$ for comparisons, and (2) a number limit $n_c$ for comparisons used in the preorder. While choosing the $i$th comparison for $\red{?R}$, only those comparisons that (1) are smaller than $n_s$, and (2) are violated on at least $1/(n_c - i + 1)$ portions of examples are considered. In this way, the search space of preorders is greatly reduced.  
\end{itemize}

We use an example to show the search procedure of \mainname. Suppose the given negative examples are $([(1, 2)], [(2, 2)])$ and $([(1, 2)], [(1, 2), (2, 2)])$ extracted in the previous example, and there are only three comparisons $c_1 = (\geq, \lambda p. |p|)$, $c_2 = (\geq, \lambda p. \textit{sumw}\ p)$ and $c_3 = (\leq, \lambda p. \textit{sumv}\ p)$ that are smaller than $n_s$. These three comparisons are violated on $1, 2$, and $0$ examples respectively, and the ranges\footnote{The range here is defined as $\textit{ma}-\textit{mi}+1$, where $\textit{ma}$ and $\textit{mi}$ are the maximum and the minimum outputs on the examples.} of their key functions on these two examples are $2, 3,$ and $4$ respectively. For simplicity, we assume the objective function is exactly the product of ranges of involved key functions.
\begin{itemize}
    \item When $n_c$ is set to $1$, \mainname considers comparisons that are violated on at least $2 / (n_c - 1 + 1) = 2$ examples. At this time, $c_2$ is the only choice and thus \mainname returns $[c_2]$ as the result.
    \item When $n_c$ is set to $2$, \mainname considers comparisons that are violated on at least $2 / (n_c - 1 + 1) = 1$ example. At this time, there are two choices $c_1$ and $c_2$. First, because $[c_2]$ satisfies all examples, \mainname updates the upper bound to its objective value, which is $3$. 
    
    Then, because $c_1$ is violated only on the second example, \mainname continues to find a preorder satisfying the first example. By branch-and-bound, at this time, \mainname only considers preorders with an objective value smaller than $3 / 2 = 1.5$. Because there is no comparison with a range smaller than $1.5$, \mainname returns immediately and thus takes $[c_2]$ as the result.
\end{itemize}

\newcommand{\nk}{\ensuremath{n_{\textit{keys}}}\xspace}
\newcommand{\szi}{\ensuremath{s_{\textit{st}}}\xspace}
\newcommand{\nf}{\ensuremath{n_{\textit{frags}}(\szi)}\xspace}
\newcommand{\ns}{\ensuremath{n_{\textit{sols}}}\xspace}
\newcommand{\szf}{\ensuremath{s_{\textit{frag}}}\xspace}
\newcommand{\szs}{\ensuremath{s_{\textit{sol}}}\xspace}
\newcommand{\Ta}{\ensuremath{T_{\textit{ana}}}\xspace}
\newcommand{\Tc}{\ensuremath{T_{\textit{cata}}}\xspace}  

\subsection{\plusname: Improving Thinning via Three Supplementary Rules}
For the second shortage, to improve the result of thinning, we analyze the factors that affect the performance of memoizing generator $g$ in Figure \ref{figure:01spec}. The time complexity is as follows.
\[
    O\left(\nk \left(\szi+ \ns \szs\right)\right)
\] 
Here \nk denotes the number of keys that key function $m$ (introduced in Figure \ref{figure:memoization}) could possibly return, \ns denotes the maximum number of partial solutions returned by a recursive call, \szi denotes the size of the search state, and \szs denote the size of a partial solution.

The execution time of memoizing $g$ is the product of the number of invocations and the execution time of a single invocation excluding the recursive call. The former is further confined by \nk. The latter consists of the execution time of processing the input and the execution time of producing the solution, where there are $\ns$ solutions, and each takes $O(\szs)$ to process. Though this analysis is specific to our example, we prove that in the general case, with some assumptions, optimizing these four factors is enough to generate an efficient memoization algorithm (Theorem \ref{theorem:efficiency}).

The main shortage of thinning is that it focuses only on $n_{\textit{sols}}$ while remaining the other three factors unchanged. In our example, though thinning optimizes $n_{\textit{sols}}$ from $O(2^n)$ to $O(n)$, both $s_{\textit{st}}$ and $s_{\textit{sol}}$ remains $O(n)$ in the resulting program and leads to the gap between the result of thinning, which runs in $O(n^3)$-time, and the standard dynamic programming algorithm for 0/1 knapsack.

Motivated by the above analysis, we propose three supplementary rules to optimize the other three factors, and propose solver \plusname, which automates and integrates these rules into \mainname. For 0/1 knapsack, our rules reduces $s_{\textit{sol}}$ and $s_{\textit{st}}$ to $O(1)$, keeps $n_{\textit{keys}}$ unchanged as $O(n)$, and thus \plusname can automatically generate an $O(n^2)$-time memoization algorithm. 

The procedure of applying these rules is listed in Figure \ref{fig:overview}. For each rule, \textit{Initial Program} shows the input for each rule where $(g_1, o_1, m_1)$ is the program generated by thinning, \textit{Intermediate Program} shows the transformation result where red variables represent the unknown functions required by the rule, \textit{Examples} lists concrete examples for the unknown functions extracted from $xs = [(3, 3), (2, 2), (1, 2)]$ and $W = 4$, and \textit{Synthesis Result} shows the functions found by \plusname.

\begin{figure}[!t]
    \centering
    \includegraphics[width=1\textwidth]{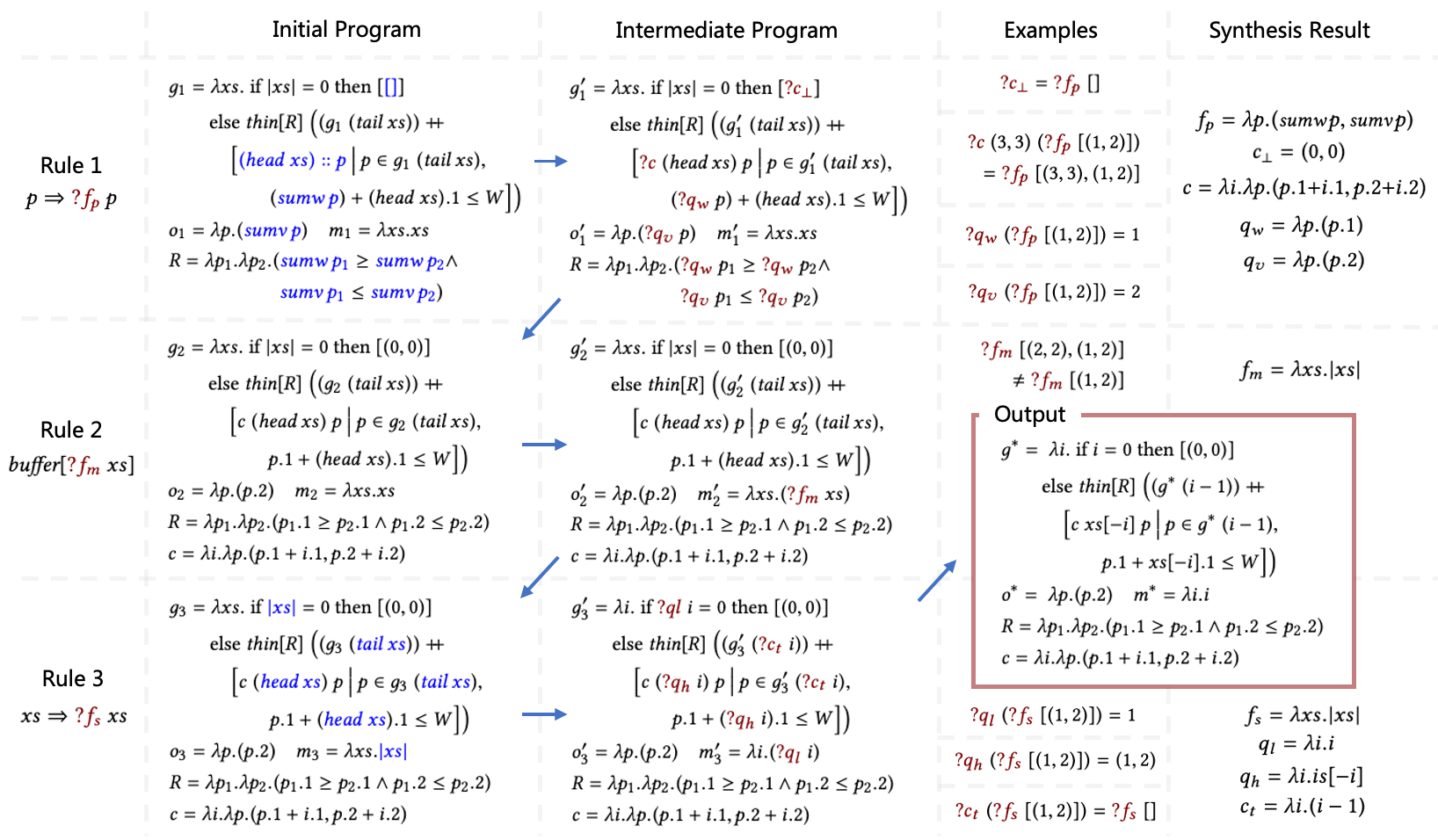}
    \caption{The procedure of applying the three supplementary rules to improve the result of thinning for 0/1 knapsack. For simplicity, we omit the first parameter of $g$ in this Figure.}
    \label{fig:overview}
\end{figure}

\smallskip
\noindent \textbf{Rule 1}. The first rule optimizes \szs, the size of partial solutions. Note that the list representation of \textit{thinning} is unnecessarily complex in $(g_1, o_1, m_1)$. To run this program, only the weight sum and the value sum of each partial solution matter. Therefore, the main idea of this rule here is to replace the representation of a solution from a list with a more compact representation, which includes only necessary information for calculating the weight sum and the value sum\footnote{Note that after changing the representation, the optimal solution in the original form can still be extracted from the optimized program. In a nutshell, one can trace back the calculation that leads to the optimal solution in the optimized program and recover the optimal solution in the original by repeating the calculation in the input program. This is a standard technique, and thus we omit it throughout this paper.}.

Rule 1 uses a \emph{converting function} $\red{?f_p}$ to convert the representation of partial solutions. It constructs an intermediate program such that each partial solution $p$ generated by the input program is also generated by the intermediate one as the output of $\red{?f_p}$. In this procedure, because the type of partial solutions is changed, those functions in the input program that access partial solutions (either takes a solution as an input or constructs a solution as the output) should be replaced correspondingly. 

In Figure \ref{fig:overview}, we mark the four functions that access partial solutions in $(g_1, m_1, o_1)$ as \blue{blue}, where $[]$ and $\lambda x. \lambda p. (x :: p)$ construct partial solutions, and $\lambda p. (\textit{sumw}\ p)$ and $\lambda p. (\textit{sumv}\ p)$ extracts information from partial solutions. Rule 1 replaces them with unknown functions $\red{?c_{\bot}}$, $\red{?c}$, $\red{?q_w}$ and $\red{?q_v}$ respectively and constructs the intermediate program $(g_1', o_1', m_1')$. In Section \ref{subsection:rule1}, we show that this transformation can be done by traversing on the AST of the hylomorphism.

\smallskip
\plusname completes the intermediate program by synthesizing functions $\red{?f_p}$, $\red{?c_{\bot}}$, $\red{?c}$, $\red{?q_w}$ and $\red{?q_v}$ via programming-by-example. To extract examples for these functions, \plusname utilizes the correspondence between the executions of the input program and the intermediate program. Given a concrete instance of 0/1 knapsack, \plusname traces the execution of the input program $(g_1, o_1, m_1)$. Each time when a solution-related function is invoked, there must be an invocation of the corresponding unknown function where each involved partial solution $p$ is replaced with the new representation $\red{?f_p}\ p$. Such an invocation is recorded as an example for synthesis.

For example, let us consider the invocation of $g_1$ with state $xs =[(3, 3), (2, 2), (1, 2)]$ and capacity $W=4$. We highlight two invocation of solution-related functions as follows.
\begin{itemize}
    \item On partial solution $[(1, 2)] \in g_1\ (\textit{tail}\ xs)$, $\lambda p. (\textit{sumv}\ p)$ is invoked, and the result is $1$. By the correspondence, there should be an invocation of $\red{?q_w}$ on the new representation of $[(1, 2)]$, i.e., $\red{?f_p}\ [(1, 2)]$, that outputs $1$. Therefore, example $\red{?q_w}\ (\red{?f_p}\ [(1, 2)]) = 1$ is obtained.
    \item On item $(3, 3)$ and partial solution $[(1, 2)]$, $\lambda x. \lambda p. (x :: p)$ is invoked, and a new partial solution $[(3, 3), (1, 2)]$ is constructed. Therefore, according to the correspondence, example $\red{?c}\ (3, 3)\ (\red{?f_p}\ [(1, 2)]) = \red{?f_p}\ [(3, 3), (1, 2)]$ is obtained.
\end{itemize}
Several extracted examples for the other functions can be found in Figure \ref{fig:overview}.

\smallskip
The remaining task for applying Rule 1 is to synthesize from the following specification, where $E_c, E_w$, and $E_v$ are the sets of extracted examples.
\begin{gather}
    \red{?c_{\bot}} = \red{?f_p}\ [] \qquad \forall (x, p) \in E_c, \red{?c} \ x\ (\red{?f_p}) = \red{?f_p}\ (x :: p) \nonumber\\
    \forall p \in E_w, \red{?q_w}\ (\red{?f_p}\ p) = \textit{sumw}\ p \qquad \forall p \in E_v, \red{?q_v}\ (\red{?f_p}\ p) = \textit{sumv}\ p \label{example:rule1-spec}
\end{gather}

\plusname reduces this task to \textit{lifting problem} and \textit{partial lifting problem}, two kinds of synthesis tasks studied by \citet{lifting}, and synthesizes by invoking an efficient synthesizer \textit{AutoLifter} for these tasks. Because the formal definitions of these tasks involve concepts in the category theory, we leave them to Section \ref{subsection:lifting} and use two examples to show the reduction made by \plusname.
\begin{itemize}
    \item Given functions $q, h$, $m = \lambda x. \lambda p. (x::p)$, and a set $E$ of examples, synthesizing $\red{?f}$ and $\red{?c}$ from the following equation is an instance $\mathsf{LP}(\{m\}, q, h, E)$ of the lifting problem.
    $$\forall (x, p) \in E, (q\ (m\ x\ p), \red{?f}\ (m\ x\ p)) = \red{?c}\ x\ (h\ p, \red{?f}\ p) 
    $$
    Clearly, the second formula in Example \ref{example:rule1-spec} is equivalent to $\mathsf{LP}(\{\lambda x. \lambda p. (x::p)\}, \textit{null}, \textit{null}, E_c)$, where $\textit{null}$ represents a dummy program returning nothing.
    \item Given functions $q, h$, $m = \lambda p. p$, and a set $E$ of examples, synthesizing $\red{?f}$ and $\red{?c}$ from the following equation is an instance $\mathsf{PLP}(\{m\}, q, h, E)$ of the partial lifting problem.
    $$\forall (x, p) \in E, q\ (m\ x\ p) = \red{?c}\ x\ (h\ p, \red{?f}\ p) 
    $$
    Clearly, third formula in Example \ref{example:rule1-spec} is equivalent to $\mathsf{PLP}(\{\lambda p. p\}, \lambda p. \textit{sumw}\ p, \textit{null}, E_w)$. 
\end{itemize}
Because in the task of applying Rule 1, function $\red{?f_p}$ is shared in all specifications, there are some details remaining on merging the results of the reduced tasks, which are left to Section \ref{subsection:rule1}.

\textit{AutoLifter} uses grammars to guarantee the efficiency of the synthesis results. Under its default setting, $\red{?f_p}$ is synthesized from a grammar including only polynomial-time programs that output only tuples of scalar values, and other functions are synthesized from a grammar including only constant-time operators for scalar values. In this way, the time complexities of all synthesis results except $\red{?f_p}$ are guaranteed to be $O(1)$. For Rule 1, because $\red{?f_p}$ is never invoked in the optimized program, such a guarantee provided by \textit{AutoLifter} already ensures the efficiency of the result.

\noindent \textbf{Rule 2}. The second rule optimizes \nk, the number of keys that $m$ possibly returns. Though the trivial function $m = \lambda xs. xs$ is already efficient in our example, in general, the search state may record too much information such that reusing results only for exactly the same search state is inefficient. An example of this case can be found in Section \ref{section:hylomorphism}, which is shown as Figure \ref{fig:p2}. To improve this point, Rule 2 replaces the original key function with a more compact one $\red{?f_m}$ and thus lets the memoized results be reused between different search states. 

To automatically applying Rule 2, \plusname synthesizes $\red{?f_m}$ from examples. To ensure the correctness of the transformation result (in our example, $(g_2', o_2', m_2')$ in Figure \ref{fig:overview}), \plusname requires $\red{?f_m}$ to assign different keys to search states with different outputs. At this time, negative example $(a,b)$ requiring that $\red{?f_m}\ a \neq \red{?f_m}\ b$ can be extracted from the execution of the input program.

For example, let us consider the invocation of $g_2$ in Figure \ref{fig:overview} with item list $xs = [(3, 3), (2, 2), (1, 2)]$ and capacity. All invocations of $g_2$ and their results are listed as follows.
\begin{gather*}
g_2\ [] = [(0, 0)] \qquad g_2\ [(1, 2)] = [(0, 0), (1, 2)] \qquad g_2\ [(2, 2), (1, 2)] = [(0,0), (1, 2), (3, 4)] \\
g_2\ [(3, 3), (2, 2), (1, 2)] = [(0,0), (1, 2), (3, 4), (4, 5)]
\end{gather*}
Because the outputs of $g_2$ are pairwise different on the four states, the outputs of $\red{?f_m}$ must be pairwise different on $[], [(1,2)], [(2, 2), (1, 2)],$ and $[(3, 3), (2, 2), (1, 2)]$, which leads to $6$ negative examples in the form of $\red{?f_m}\ a \neq \red{?f_m}\ b$ for function $\red{?f_m}$.

\smallskip 
To focus on effective candidates of $\red{?f_m}$, \plusname considers only those functions compressed the search state to a tuple of scalar values. At this time, $\red{?f_m}$ can be regarded as a tuple of key functions.
$$
\red{?f_m}\ s = (\red{?key_1}\ s, \dots, \red{?key_n}\ s)
$$

\plusname synthesizes $\red{?f_m}$ from examples by reducing it to the synthesis task for thinning in \mainname. Note that key function $\red{?f_m}$ can be regarded as an equivalence relation $\red{?R_m}$ over search states, where $a\red{?R_m}b$ is defined as $\red{?f_m}\ a = \red{?f_m}\ b$. The following shows the expanded form of $\red{?R_m}$.
$$
a \red{?R_m} b \iff \wedge_i (\red{?key_i}\ a) = (\red{?key_i}\ b)
$$

First, the form of $\red{?R_m}$ matches the form of preorders considered by \mainname. Second, the number of different keys returned by $\red{?f_m}$ is bounded by the product of the ranges of key functions, which matches the objective function used by \mainname. Therefore, an efficient $\red{?f_m}$ can be directly synthesized by invoking the solver in \mainname.

For program $(g_2', o_2', m_2')$, \mainname synthesizes $\red{?f_m}$ as $\lambda xs. |xs|$. Though the value \nk does not change, the synthesized key function simplifies the information on search states required for memorization, and thus make it possible to simplify the representation of search states later.

\smallskip
\noindent \textbf{Rule 3}. The step optimizes \szi, the size of search states. The procedure of applying this rule is almost the same with Rule 2. First, \plusname uses a converting function $\red{?f_s}$ to convert the representation of search states, and generates an intermediate program $(g_3', o_3', m_3')$ by replacing all state-related functions $(g_3, o_3, m_3)$ with unknown functions. Second, \plusname extracts examples by tracing the execution of $(g_3, o_3, m_3)$ on concrete instances of 0/1 knapsack. Last, \plusname synthesizes unknown functions from examples by invoking \textit{AutoLifter}.

For program $(g_3', o_3', m_3')$, \plusname synthesizes $\red{?f_s}$ as $\lambda xs. |xs|$ and uses the length of the item list to represent a search state. In this way, $\szi$ is reduced from $O(n)$ to $O(1)$.

\smallskip
\noindent \textbf{Result}. The result of \plusname for 0/1 knapsack is shown as $(g^*, o^*, m^*)$ in Figure \ref{fig:overview}. In this program, $\szi=\szs=O(1)$, $\nk=\ns=O(n)$, and thus the time-complexity is reduced to $O(n^2)$.
\section{Preliminaries} \label{section:background}
To operate functions, the following four operators $\circ, +, \times$, and $\triangle$ will be used in our paper.
\begin{gather*}
    (f_1 \circ f_2)\ x\coloneqq f_1\ (f_2\ x) \quad (f_1 + f_2)\ (i, x) \coloneqq f_i\ x, i \in \{1, 2\} \\
    (f_1 \times f_2)\ (x, y) \coloneqq (f_1\ x, f_2\ y) \quad (f_1 \triangle f_2) \ x\coloneqq (f_1\ x, f_2\ x)
\end{gather*}
\subsection{Categorical Functors}
\textit{Functor} is an important concept in category theory. A category consists of a set of objects, denoted by uppercase letters such as $A$, and a set of arrows between objects, denoted by lowercase letters such as $f$. In this paper, we focus on category \textbf{Fun}, where an object is a set and an arrow from object $A$ to object $B$ is a total function from set $A$ to set $B$. 
In a category, a functor $\F$ maps objects to objects, arrows to arrows, and keeps identity and composition. 
$$
\F id_A = id_{\F_A} \qquad \F(f \circ g) = \F f \circ \F g
$$
where $\textit{id}_A$ represents the identity function on set $A$. Intuitively, a functor can be regarded as a higher-order function, which constructs new functions from existing functions.

A functor is a \textit{polynomial functor} if it is constructed by identity functor $\mathsf{I}$, constant functors $!A$, and bifunctors $\times, +$. Their definitions are shown below.
\begin{gather*}
    \mathsf{I}A \coloneqq A\ \ \ \mathsf{I}f \coloneqq f\quad (!A)B \coloneqq A\ \ \ (!A)f \coloneqq id_A \quad 
    (\F_1 \times \F_2)A \coloneqq \F_1A \times \F_2A \ \ \ (\F_1 \times \F_2)f \coloneqq \F_1f \times \F_2 f \\
    (\F_1 + \F_2)A \coloneqq (\{1\} \times \F_1A) \cup(\{2\} \times \F_2A) \ \ \
    (\F_1 + \F_2)f \coloneqq \F_1 f + \F_2 f
\end{gather*}
where $A \times B$ represents the Cartesian product of objects $A$ and $B$.

In this paper, when using symbol $\F$, we inherently assume that $\F$ is a polynomial functor. Besides, we also use the power functor $\mathsf{P}$ to express operations related to power sets.
\begin{align*}
    \mathsf P A \coloneqq \{s\ |\ s \subseteq A\} \quad \mathsf P f\ s \coloneqq \{f\ a\ |\ a \in s\}
\end{align*}

\subsection{Generator, Memoization, and Hylomorphism} \label{section:hylomorphism}
The concept of \textit{relational hylomorphism} is originally defined on another category namely \textbf{Rel}. For simplicity, in this paper, we introduce it as its simplified counterpart in category \textbf{Fun}.

\begin{definition}[Recursive Generator] \label{def:rec} Given arrows $\phi: \mathsf {PF}A \rightarrow \mathsf P A$ and $\psi: B \rightarrow \mathsf P\F B$, recursive generator $rg(\phi, \psi)_{\F}: A \rightarrow \mathsf PB$ is the smallest solution of the following equation, where arrow $r_1$ is smaller than $r_2$ if $\forall i, r_1\ i \subseteq r_2\ i$.
    $$
    rg(\phi, \psi)_{\F} = \phi \circ \textit{cup} \circ \mathsf P (\textit{car}[\mathsf F] \circ \mathsf {F}rg(\phi, \psi)_{\F}) \circ \psi
    $$
In this equation, $\textit{cup}: \mathsf{PP}A \rightarrow \mathsf{P}A$ is defined as $\textit{cup}\ x \coloneqq \cup_{s \in x} s$, which unions all sets in a set of sets, and $\textit{car}[\F]: \mathsf{FP}A \rightarrow \mathsf{PF}A$ is defined as the following, which is similar to the Cartesian product.
\begin{gather*}
\textit{car}[\mathsf I]\ x \coloneqq x \quad
\textit{car}[\F_1 \times \F_2]\ (x_1, x_2) \coloneqq (\textit{car}[\F_1]\ x_1) \times (\textit{car}[\F_2]\ x_2) \\
\textit{car}[!A]\ x \coloneqq x \quad\textit{car}[\F_1 + \F_2]\ (i, x) \coloneqq \{(i, a)\ |\ a \in x\}
\end{gather*}
\end{definition}

The concept of recursive generators corresponds to recursive programs in \textbf{Rel}. As discussed in Section \ref{section:overview}, a recursive generator can be naturally memoized via a key function $m$. In the remainder of this paper, we use $r^m$ to denote the result of memoization. 

In this paper, recursive generator $rg(\phi, \psi)_{\F}$ is invoked to generate valid solutions for a given COP. Similar to the discussion in Section \ref{section:overview}, we denote the input of $rg(\phi, \psi)_{\F}$ as a \textit{search state}, an element in the output of $rg(\phi, \psi)_{\F}$ as a \textit{solution}, and an element in the output of some invocation involved in the recursive definition as \textit{partial solution}.

\begin{definition} [Relational Hylomorphism]\label{def:hylo} Given two arrows $\phi: \mathsf F A\rightarrow \mathsf P A$ and $\psi: B \rightarrow \mathsf P\F B$, relational hylomorphism $\hylo{\phi, \psi}_{\F}: B \rightarrow \mathsf P A$ is defined as $rg(\textit{cup} \circ \mathsf P\phi, \psi)$.
\end{definition}

Compared to a general recursive generator, a hylomorphism assumes the independence while constructing each solution. For hylomorphism $\hylo{\phi, \psi}_{\F}$, given the set including all sub-results of the recursions, i.e., $(\textit{cup} \circ \mathsf P (\textit{car}[\mathsf F] \circ \mathsf {F} \hylo{\phi, \psi}_{\F}) \circ \psi)$, $\hylo{\phi, \psi}_{\F}$ independently constructs solutions for each result via $\mathsf P \phi$, and then merges all solutions via $\textit{cup}$.

\begin{figure*}[t]
    \begin{minipage}{.4\linewidth}
      \centering
          {\small 
          \begin{align*}
          &\textit{prog} = (\hylo{\phi, \psi}_{\mathsf F}, o), \text{ where }\\
          &\quad\F = !\texttt{Unit} + \mathsf I +  ((!\texttt{Int} \times !\texttt{Int}) \times \mathsf I) \\
          &\quad \psi = \lambda xs. \text{ if }|xs| = 0\text{ then collect } (1, \textit{unit}); \\
          &\quad \quad \text{else }\{\text{collect }(2, \textit{tail}\ xs);\\\
          & \quad \quad \quad \text{collect } (3, (\textit{head}\ xs, \textit{tail}\ xs));\}\\
          &\quad \phi = \lambda (\textit{tag}, p). \text{ if }\textit{tag} = 1\text{ then collect } []; \\
          &\quad \quad \text{else if } tag = 2 \text{ then }\text{collect }p.2; \\
          &\quad \quad \text{else if }(\textit{sumw}\ p.2) + p.1.1 \leq W\text{ then} \\
          & \quad \quad \quad \quad \text{collect } (p.1 :: p.2); \\
          &\quad o = \lambda p. \textit{sumv}\ p
          \end{align*}
          }
          \vspace{-17pt}
          \caption{The program corresponding to $(g, o)$.}
          \label{fig:p1}
          {\small 
          \begin{align*}
          &\textit{prog}_2 = (\hylo{\phi, \psi}_{\mathsf F}, o), \text{ where }\\
          &\quad\F = !\texttt{Unit} + ((!\texttt{Int} \times !\texttt{Int}) \times \mathsf I) \\
          &\quad \psi = \lambda xs. \text{ if }|xs| = 0\text{ then collect } (1, \textit{unit}); \\
          &\quad \quad \text{else collect } (2, (\textit{head}\ xs, \textit{tail}\ xs)); \\
          &\quad \phi = \lambda (\textit{tag}, p). \text{ if }\textit{tag} = 1\text{ then collect } []; \\
          &\quad \quad \text{else }\{\text{collect }(p.1::p.2);\ \text{collect }p.2;\} \\
          & \quad o = \lambda p. (\textit{sumw}\ p \leq W)\ ?\ (\textit{sumv}\ p) :-\infty;
          \end{align*}
          }
          \vspace{-17pt}
          \caption{A valid program for 0/1 knapsack in $\mathcal L_H$. In this program, $-\infty$ is a small enough integer used to exclude invalid programs.}
          \label{fig:p3}
    \end{minipage} \quad
      \begin{minipage}{.48\linewidth}
        \centering
            {\small
                \begin{tabular}{cccl} 
                    \toprule 
                    Program& $\mathbb P$& $\rightarrow$ & $(\mathbb H, \mathbb E)$ \\
                    Hylomorphism & $\mathbb H$ & $\rightarrow $ & $\hylo{\lambda x. \mathbb S, \lambda x. \mathbb S}_{\mathbb F}$ \\
                    Functor& $\mathbb F$& $\rightarrow$& $\mathsf I\ |\ !\texttt{Unit}\ |\ !\texttt{Int}\ |\ \mathbb F \times \mathbb F$\\ 
                    & & $|$ & $\mathbb F + \mathbb F $ \\
                    Statement & $\mathbb S$ & $\rightarrow$ & $\text{skip};\ |\ \mathbb S\ \mathbb S \ |\ \text{collect }\mathbb E;$ \\ 
                    & & $|$ & $\text{if }\mathbb E\text{ then }\mathbb S\text{ else }\mathbb S$ \\
                    & & $|$ & $\text{foreach }x \in [\mathbb E, \mathbb E] \text{ in }\mathbb S$ \\
                    Expression & $\mathbb E$ & $\rightarrow$ & $x \ |\ \text{const}\ |\ \lambda x. \mathbb E $ \\
                    & & $|$ & $\oplus\ \mathbb E\dots\mathbb E$\\
                    \bottomrule
            \end{tabular}}
            \caption{The syntax of language $\mathcal L_H$, where $x$ represents a variable, $\oplus$ represents a black-box operator.} 
            \label{fig:syntax}
            \vspace{-7pt}
            \centering
            {\small 
            \begin{align*}
            &\textit{prog}_3 = (\hylo{\phi, \psi}_{\mathsf F}, o), \text{ where }\\
            &\quad\F = !\texttt{Unit} + \mathsf I + ((!\texttt{Int} \times !\texttt{Int}) \times \mathsf I) \\
            &\quad \psi = \lambda (xs, p). \text{ if }|xs| = 0\text{ then collect } (1, \textit{unit}); \\
            &\quad \quad \text{else }\{ \text{collect }(2, (\textit{tail}\ xs, p));\\
            & \quad \quad \quad \text{if } (\textit{sumw}\ p) + (\textit{head}\ xs).1 \leq W \text{ then} \\
            & \quad \quad \quad \text{collect }(3, (\textit{head}\ xs, (\textit{tail}\ xs, (\textit{head}\ xs)::p))); \}\\
            &\quad \phi = \lambda (\textit{tag}, p). \text{if }\textit{tag} = 1\text{ then collect }[]; \\ 
            &\quad \quad \text{else if }\textit{tag} = 2 \text{ then collect }p; \\
            & \quad \quad \text{ else collect }(p.1::p.2);\\
            &\quad o = \lambda p. \textit{sumv}\ p
            \end{align*}
            }
            \vspace{-17pt}
            \caption{A valid program for 0/1 knapsack in $\mathcal L_H$. The search state here is $(xs, p)$, where $xs$ is the list of remaining items, $p$ is the list of selected items.}
            \label{fig:p2}
      \end{minipage} 
\end{figure*}

Figure \ref{fig:p1} shows a program corresponding to the input program $(g, o)$ specified in Figure \ref{figure:01spec}, where $g$ is expressed by hylomorphism $\hylo{\phi, \psi}_{\F}$. Because both $\phi, \psi$ return a set, we use $|\textit{collect}|$ instead of $|\textit{return}|$ to express their outputs, where $|\textit{collect}\ e|$ inserts the value of $e$ to the resulting set.

Relational hylomorphisms are natural for specifying COPs, where $\psi$ and $\phi$ specify the recursive structure and the construction of solutions respectively. However, relational hylomorphism is not enough to express optimizations (e.g., thinning), where the construction of solutions is no longer independent due to the optimization. Therefore, \mainname takes a relational hylomorphism as the input but expresses the internal optimized programs via recursive generators.




\subsection{Programs} \label{subsection:program}
In \mainname and \plusname, a program is represented by a pair $(g, o)$.  Given an instance $i$ of the COP, the output of $(g, o)$ is equal to $\textit{argmax}\ o\ (g\ i)$. In terms of solving COPs, two programs are equivalent on an instance if they achieve the same objective value.

\begin{definition}[Equivalence] Two programs $(g_1, o_1)$ and $(g_2, o_2)$ are equivalence on instance $i$, denoted as $(g_1,o_1) \sim_{i} (g_2,o_2)$, if $\max(g_1\ p), p \in (g_2\ i) = \max(o_2\ p), p \in (g_2\ i)$.
\end{definition} 
This concept can be naturally extended to a set of instances. Program $(g_1, o_1)$ and $(g_2, o_2)$ are equivalence on a set $I$ of instances, denoted as $(g_1, o_1) \sim_{I} (g_2, o_2)$, if $\forall i \in I, (g_1, o_1) \sim_i (g_2, o_2)$. 

In this paper, we provide a simple language $\mathcal L_H$ for specifying COPs via relational hylomorphisms. The syntax of this language is shown as Figure \ref{fig:syntax}, and three different programs in $\mathcal L_H$ for describing \textit{0/1 knapsack} are shown as Figure \ref{fig:p1}, \ref{fig:p3} and \ref{fig:p2}, where $|\textit{if }\mathbb E \textit{ then } \mathbb S|$ is a sugar of $|\textit{if }\mathbb E \textit{ then }\mathbb S \textit{ else skip};|$, and $|\lambda (x_1, x_2). \mathbb S|$ is a sugar of extracting components in the input for $|\lambda x. \mathbb S|$. 




\subsection{Thinning} \label{section:thinning}
The definition of thinning is based on \textit{preorders}. A preorder on object $A$ is a relation that is reflexive ($\forall a \in A, aRa$) and transitive $(\forall a,b,c \in A, aRb \wedge bRc \rightarrow aRc)$. 

\begin{definition} \label{def:thin}Given a preorder $R$ on $A$, $\textit{thin}[R]: \mathsf PA \rightarrow \mathsf PA$ is an arrow such that for any set $s \subseteq A$, $\textit{thin}[R]\ s$ is the smallest subset of $s$ satisfying $\forall a \in s, \exists b \in \textit{thin}[R]\ s, aRb$.
\end{definition}

In this paper, we focus on a special case of $\textit{thin}[R]$ where $R$ is the conjunction of comparisons on several key functions. We denote such a preorder as a \textit{keyword preorder}.
\begin{definition}[Keyword Preorder] A keyword preorder $R$ of comparisons $\{(op_i, k_i)\}$ on $A$ is defined as $aRb \iff \wedge_{i} (k_i\ a)op_i (k_i\ b)$, where $k_i$ is an arrow from $A$ to $\texttt{Int}$ and $op_i \in \{\leq, =, \geq\}$.
\end{definition}

Given a keyword preorder $R$, Theorem \ref{theorem:thinning} shows that the size of the set returned by $\textit{thin}[R]$ and the time cost of $\textit{thin}[R]$ can both be bounded by the ranges of the key functions involved in $R$.

\begin{theorem} \label{theorem:thinning}Given a keyword preorder $R$ of $\{(op_i, k_i)\}$, define function $N_R(S)$ as the following, where $\textit{range}(k, S)$ is the range of $k$ on $S$, i.e., $\max_{a \in S}(k\ a) - \min_{a \in S}(k \ a)$, and $\max_1(S)$ returns the largest element in $S$ with default value $1$.
$$
N_R(S) \coloneqq \left( \prod_{i}\textit{range}(k_i, S)\right) \bigg/ \max_i\big(\textit{range}(k_i, S)\ \big|\ op_i \in \{\leq, \geq\}\big)  
$$
\begin{itemize}[leftmargin=*]
\item For any set $S$, $|\textit{thin}[R]\ S| \leq N_R(S)$.
\item There is an implementation of $\textit{thin}[R]$ with time complexity $O(N_R(S)\textit{size}(R) + T_R(S))$, where $S$ is the input set, $\textit{size}(R)$ is the number of comparisons in $R$, $T_R(s)$ is the time complexity of evaluating all key functions in $R$ for all elements in $S$.
\end{itemize}
\end{theorem}

Due to the space limit, we omit the proofs to the theorems and move them to the appendix.

\subsection{Thinning Theorem}
The thinning theorem proposed by \citet{DBLP:books/daglib/0096998} shows that $\textit{thin}[R]$ can be used to derive efficient memoization algorithms for COPs. In this paper, we use the following variant of the thinning theorem, which generalizes the original one to all relational hylomorphisms.

Given program $(h = \hylo{\phi, \psi}_{\F}, o)$ and a search state $s$, invoking $h$ on $s$ without memoization generates a search tree where each vertex corresponds to a search state. We introduce two notations $T_h$ and $S_h$ to access the structure of this tree, where $T_h\ s$ and $S_h\ s$ represent the set including all direct children of $s$ and the set including all states in the subtree of $s$ respectively. 

Besides, we introduce relation $\twoheadrightarrow_{h, s}$ to denote the constructions of solutions. For partial solution $p \in h\ s$ and tuple $(p_1, \dots, p_k)$ of partial solutions, $(p_1, \dots, p_k) \twoheadrightarrow_{h, s} p$ holds if there is an invocation of $\phi$ where $(p_1, \dots, p_k)$ are all partial solutions used in the input and $p$ is inside the output.

\begin{theorem}[Thinning Theorem] \label{theorem:thinning-theorem} Given program $(h\!= \!\hylo{\phi, \psi}_{\F},o)$ and preorder $R$, for any instance $i$,  $(rg(\textit{thin}[R] \circ \textit{cup} \circ \mathsf P \phi, \psi)_{\F}, o) \sim_i (h, o)$ if the following two conditions are satisfied.
    \begin{enumerate}
    \item $\forall s \in S_h\ i, \forall p_1, p_2 \in h\ s, p_1Rp_2 \rightarrow (o\ p_1 \leq o\ p_2)$.
    \item $\forall s \in S_h\ i, \forall \overline{p_1} = (p_{1,1}, \dots, p_{1,k}), \overline{p_2} = (p_{2,1}, \dots, p_{2,k})$, where $p_{1,i}$ and $p_{2,i}$ are partial solutions of the same search state for all $i \in [1, k]$, the following formula is always satisfied. 
    \begin{align}
    \bigwedge_{i=1}^k p_{1,i}Rp_{2,i} \rightarrow \forall p_1', \bigg(\overline{p_1} \twoheadrightarrow_{h,s} p_1' \rightarrow \exists p_2', \big(\overline{p_2}\twoheadrightarrow_{h,s} p_2' \wedge p_1' R p_2'\big)\bigg) \label{formula:thinning2}
    \end{align}
    \end{enumerate}
\end{theorem}

\subsection{Lifting and Partial Lifting} \label{subsection:lifting}
Lifting problems and partial lifting problems are synthesis tasks studied by \citet{lifting}, which are generalized from the synthesis task for automated parallelization.

\begin{definition} \label{def:lifting}Given arrows $p, h$ starting from object $A$, a set $M$ including $n$ arrows $m_i: \mathsf F_{m_i}A \rightarrow A$, and an example space $E$ attaching a set of examples $E[m_i]$ to each $m_i \in M$, lifting problem $\mathsf{LP}(M, p, h, E)$ and partial lifting problem $\mathsf{PLP}(M,p,h, E)$ are to find $\red{?f}, \red{?c_1}, \dots, \red{?c_n}$ such that Equation \ref{formula:lp} and \ref{formula:plp} are satisfied for all $m_i \in M$, respectively.
    \begin{gather}
    \forall e \in E[m_i], ((p \triangle \red{?f}) \circ m_i)\ e = (\red{?c_i} \circ \mathsf F_{m_i}(h \triangle \red{?f}))\ e \label{formula:lp}\\
    \forall e \in E[m_i], (p \circ m_i)\ e = (\red{?c_i} \circ \mathsf F_{m_i}(h \triangle \red{?f}))\ e \label{formula:plp}
    \end{gather}
\end{definition}

\textit{AutoLifter}~\cite{lifting} is an efficient synthesizer for these two tasks, and guarantees that the time complexities of $\red{?f}$ and $\red{?c}$ are polynomial-time and constant-time respectively.
\section{\mainname: Automating Thinning} \label{section:main}
Given program $(h = \hylo{\phi, \psi}_{\mathsf F}, o)$ and a set of instances $I$, \mainname generates a memoization algorithm by applying thinning to $(h, o)$. Concretely, \mainname synthesizes a keyword preorder $\red{?R}$ satisfying Theorem \ref{theorem:thinning-theorem} for all instance $i \in I$, and returns the following program.
\begin{align}\textit{prog}_1 = (rg(\textit{thin}[\red{?R}] \circ \textit{cup} \circ \mathsf P \phi, \psi)_{\F}^m, o)\label{form:thinning}\end{align}
where $m$ is the trivial key function $\lambda s. s$ for memoization. 

\subsection{Synthesis Task} \label{section:main-task}
To generate an efficient memoization algorithm, \mainname needs to find a keyword preorder $\red{?R}$ such that the result of applying thinning with $\red{?R}$ is correct and efficient.

For correctness, $\red{?R}$ must satisfy the two conditions provided by Theorem \ref{theorem:thinning-theorem}. \mainname ensures the first condition by requiring $\red{?R}$ to include comparison $(\leq, o)$, and considers the following equivalent form of Formula \ref{formula:thinning2} in the second condition. 
\begin{align}
    \exists p_1', \bigg(\overline{p_1} \twoheadrightarrow_{h,s} p_1' \wedge \forall p_2', \big(\neg \overline{p_2}\twoheadrightarrow_{h,s} p_2' \vee \neg p_1' R p_2'\big)\bigg) \rightarrow \bigvee_{i=1}^k \neg p_{1,i} R p_{2, i} \label{formula:thinning-eq}
\end{align}

Given a concrete preorder $R$ and an instance $i$, we denote pair $(\overline{p_1}, \overline{p_2})$ as a counter-example for $R$ on instance $i$ if the Formula \ref{formula:thinning-eq} is violated after substituting into $\overline{p_1}$ and $\overline{p_2}$. Let $CE(R, i)$ be the set of all counter-examples of $R$ on instance $i$. Then finding a correct preorder $\red{?R}$ for thinning on instance $i$ is equivalent to finding $\red{?R}$ such that $CE(\red{?R}, i)$ is empty.

For two keyword preorders $R_1$ and $R_2$ where comparisons in $R_1$ form a subset of those in $R_2$, Lemma \ref{lemma:monoce} relates the sets of counter-examples of $R_1$ and $R_2$ with the extra comparisons in $R_2$.

\begin{lemma}\label{lemma:monoce} Given instance $i$, for any two keyword preorders $R_1, R_2$ where all comparisons in $R_1$ are included in $R_2$, the following formula is always satisfied.
    $$
    \forall (\overline{p_1}, \overline{p_2}) \in CE(R_1, i), (\overline{p_1}, \overline{p_2}) \notin CE(R_2, i) \leftrightarrow \neg \overline{p_1}(R_2/R_1) \overline{p_2}
    $$
where $R_2/R_1$ represents the keyword preorder formed by the comparisons in $R_2$ that are not used in $R_1$.
\end{lemma}

Lemma \ref{lemma:monoce} suggests an incremental synthesis scheme for $\red{?R}$. To synthesize a correct preorder by enlarging a known keyword preorder $R_1$ with an unknown one $\red{?R_2}$, $\red{?R_2}$ must \textit{satisfy} all examples $(\overline{p_1}, \overline{p_2})$ in $CE(R_1, i)$, where $\red{?R_2}$ satisfies example $(\overline{p_1}, \overline{p_2})$ is defined as $\neg \overline{p_1} \red{?R_2} \overline{p_2}$.  

For efficiency, the number of plans returned by $\textit{thin}[\red{?R}]$ and the time cost should be minimized. Guided by Theorem \ref{theorem:thinning}, \mainname optimizes the following objective function while synthesis.
\begin{equation*}
\textit{cost}(\red{?R}, I) \coloneqq N_{\red{?R}}(P), \quad P = \big\{p\big|i \in I, s \in S_h\ i, p \in h\ s\big\} 
\end{equation*}

\subsection{Synthesis Algorithm} \label{section:main-algorithm}
We start with a synthesis algorithm for a subtask where a finite set of comparisons $C$ and a size limit $n_c$ are provided. In this subtask, the search space of $\red{?R}$ is constrained to keyword preorders constructed by $(\leq, o)$ and at most $n_c$ comparisons in $C$.

\SetKwFunction{Search}{BestPreorder}
\SetKwFunction{CC}{CandidateComps}
\begin{algorithm}[t]
    \small
    \caption{Synthesizing preorder $\red{?R}$ for applying thinning.}
    \KwIn{Input program $(h, o)$, a set of instances $I$, a set of comparisons $C = \{(op_i,k_i)\}$, and a size limit $n_c$.}
    \KwOut{A keyword preorder $\red{?R}$ for $(h, o)$ that involves only comparisons in $C$.}
    \label{alg:preorder}
    \LinesNumbered
    \SetKwProg{Fn}{Function}{:}{}
    \Fn{$\Search{$R, lim, \textit{costLim}$}$}{
        $es \gets \cup_{i \in I}CE(R, i)$; \\
        \lIf{$\textit{es} = \emptyset$}{\Return $R$} 
        \lIf{$lim = 0$}{\Return $\bot$}
        $cList \gets \CC(C, lim, es)$; \\
        Sort $cList$ in the increasing order of $\textit{cost}(R \cup \{c\}, I)$; \\
        $\textit{res} \gets \bot$; \\
        \ForEach{$c \in cList$}{
            \lIf{$\textit{cost}(R \cup \{c\}, I) \geq \textit{costLim}$}{\textbf{continue}}
            $R' \gets \Search{$R \cup\{c\},lim-1, \textit{costLim}$}$; \\
            \lIf{$R' \neq \bot$}{$(\textit{res}, \textit{costLim}) \gets (R', \textit{cost}(R', I))$}
        }
        \Return $\textit{res}$;
    }
    \Return $\Search{$\{(\leq, o)\}, n_c, +\infty$}$;
\end{algorithm}

As shown in Algorithm \ref{alg:preorder}, \mainname solves this subtask via \textit{branch-and-bound}. The main function \Search decides comparisons used in $\red{?R}$ one by one with an upper bound on the cost (Lines 1-14). In each turn, a set of candidate comparisons is identified via $\CC$ (Line 5) and they are considered in the increasing order of the cost (Line 6). For each comparison $c$, if its cost is smaller than the bound (Line 9), preorders including $c$ will be considered recursively (Line 10). Results of recursions will be used to update the bound (Line 11).

The implementation of $\CC{}$ is crucial to the efficiency of Algorithm \ref{alg:preorder}. If it returns too many candidate comparisons, the search space of \Search will be too large to be explored efficiently. Based on the following lemma, $\CC{$C, lim, es$}$ in \mainname returns only those comparisons that satisfy at least $|es|/lim$ examples.

\begin{lemma}\label{lemma:candidatecmp} Given a set of instances $I$, for any two keyword preorders $R_1, R_2$ where all comparisons in $R_1$ are included in $R_2$ and $\forall i \in I, CE(R_2, i) = \emptyset$, there exists a comparison $(op, k) \in R_2/R_1$ satisfying at least $1/(|R_2| - |R_1|)$ portion of examples in $CE(R_1, I) =\cup_{i \in I} CE(R_1, i)$, i.e., 
\begin{align*}
\left|\left\{(p_1, p_2)\in CE(R_1, I)\ \big|\ \neg \big((k\ p_1)op(k\ p_2)\big)\right\}\right| \geq  |CE(R_1, I)| \big / (|R_2|-|R_1|)
\end{align*}
where $|R|$ represents the number of comparisons in keyword preorder $R$.
\end{lemma}

A direct implementation of \CC{$C, lim, es$} in Algorithm \ref{alg:preorder} (Line 5) is to evaluate all comparisons in $C$ on all examples in $es$. Such an implementation is inefficient because both $C$ and $es$ can be large. \mainname further improves this point by sampling. For each comparison in $C$, \mainname decides whether to include it in \CC{$C, lim, es$} in two steps.
\begin{itemize}
\item First, \mainname draws $lim \times n_t$ random examples from $es$, where $n_t$ is a given parameter. The comparison will be ignored if the number of samples it satisfies is smaller than $k_t = \lfloor lim/2 \rfloor$.
\item Second, the comparison is evaluated on all examples and is returned by \CC only when it satisfies at least $|es|/lim$ examples.
\end{itemize}
Specially, $k_t$ is set to $lim \times n_t$, i.e., the number of samples, when $lim$ is equal to $1$. 

By Chernoff bound, the probability for a comparison that satisfies at least $1/lim$ portion of examples in $es$ to be filtered out by the samples is at most $\exp(-n_t/8)$. Therefore, the probability for Algorithm \ref{alg:preorder} to incorrectly exclude some comparison can be controlled by $n_t$. Moreover, as we shall show later, because Algorithm \ref{alg:preorder} in \mainname is wrapped in an iteration procedure, such an error rate does not affect the completeness of \mainname, as demonstrated by Theorem \ref{theorem:complete}.

The following is some details to make Algorithm \ref{alg:preorder} useful for synthesizing $\red{?R}$ in practice. 

\noindent \textbf{Decide a finite set of comparisons}. In practice, the space of comparisons is specified by a grammar, in which the number of comparisons is infinite. Because of the difficulty of finding the optimal preorder from an infinite set of comparisons, \mainname approximates it via the principle of \textit{Occam's Razor} and prefers to use smaller comparisons to construct $\red{?R}$. 
Concretely, \mainname selects a parameter $s_c$ and takes all comparisons no larger than $s_c$ in the grammar as set $C$.

\noindent \textbf{Decide $s_c$ and $n_c$}. Because both parameters $s_c$ (the size limit of comparisons) and $n_c$ (the number limit of comparisons) are not given in practice, \mainname decides them iteratively with two parameters $s_c^*$ and $n_c^*$. In each turn, \mainname considers the subtask where $n_c = n_c^*$ and $s_c = s_c^*$. If there is no solution, both $s_c^*$ and $n_c^*$ will be increased by $1$ in the next iteration.

\smallskip
We prove the completeness of \mainname for synthesizing preorders in the following theorem.
\begin{theorem} \label{theorem:complete} Given program $(h, o)$, a set of instances $I$ and a grammar $G$ for available comparisons, if there exists a keyword preorder $R$ satisfying (1) $\forall i \in I, CE(R, i) = \emptyset$, and (2) $R$ is constructed by $(\leq, o)$ and some comparisons in $G$, \mainname must terminate and return such a keyword preorder.
\end{theorem} 
\section{\plusname: Improving Thinning via Three Rules} \label{section:plus}
To improve the memoization algorithm generated by thinning, \plusname introduce three supplementary rules on the basis of \mainname to optimize three factors ignored by thinning: (1) the size of solutions, (2) the number of keys in memoization, and (3) the size of search states. 

\subsection{Rule 1: Optimizing the Representation of Solutions} \label{subsection:rule1}
The program $\textit{prog}_1$ generated by thinning is in Form \ref{form:thinning}, where the bodies of $\phi$ and $\psi$ are statements ($\mathbb S$) in language $\mathcal L_H$ (Figure \ref{fig:syntax}). \plusname optimizes the size of partial solutions via a converting function $\red{?f_p}$, which maps partial solutions into a tuple of scalar values. 

\plusname constructs an intermediate program $\textit{prog}_1'$ that performs almost the same as $\textit{prog}_1$ except each partial solution $p$ is stored as $\red{?f_p}\ p$. The construction is done by rewriting all solution-related functions, which can be classified into the following two types according to their output type.
\begin{itemize}
\item Constructors, which constructs a new partial solution. In language $\mathcal L_H$, constructors are the sub-expressions $sp_e$ of all $|\textit{collect}\ sp_e|$ in $\phi$.
\item Queries, which returns a tuple of scalar values. In language $\mathcal L_H$, queries are all solution-related expressions outside $|\textit{collect}\ sp_e|$ in program $\textit{prog}_1$.
\end{itemize}

The following shows the content of $\textit{prog}_1'$.
\begin{equation}
\textit{prog}_1' = (rg((\textit{thin}[R']) \circ \textit{cup}\circ \mathsf P \phi', \psi)_{\F}, \red{?q[o]}) \label{form:prog2new}
\end{equation}

\SetKwFunction{ES}{RewriteS}
\SetKwFunction{EE}{RewriteE}
\SetKwFunction{Child}{Children}
\SetKwFunction{Insert}{Insert}
\SetKwFunction{Rewrite}{Clone}
 
\begin{algorithm}[t]
    \small
    \caption{ Construct $\phi'$ from $\phi$.}
    \KwIn{A program $\phi = \lambda s. \phi_s$, where $\phi_s$ is a statement in $\mathcal L_H$.}
    \KwOut{Queries $Q$, Constructors $M$, and $\phi'$ in $\textit{prog}_2$.}
    \label{alg:extract}
    \LinesNumbered
    \SetKwProg{Fn}{Function}{:}{}
    $Q \gets \emptyset; \ \ M \gets \emptyset;$ \\
    \Fn{$\EE{$p_e$}$}{
        \lIf{$|s| \notin p_e$}{\Return $p_e$}
        \lIf{$p_e = |s|$}{\Return $\bot$}
        $sp_e' \gets \{\EE{$sp_e$}\ |\ sp_e \in \Child{$p_s, \mathbb E$}\}$; \\
        \If{$\bot \in sp_e'$}{
            \If{$p_e$ output a tuple of scalar values}{
            $t_1, \dots, t_m \gets $ temporary variables in $p_e$; \\
            $Q.\Insert{$p_e$}$;\ \ \Return $|\red{?q[p_e]}\ (s, t_1, \dots, t_m)|$; 
            }
            \Return $\bot$;
        }
        \Return \Rewrite{$p_e, sp_e'$};
    }
    \Fn{$\ES{$p_s$}$}{
        \If{$p_2 = |\text{collect}\ sp_e|$}{
            $t_1, \dots, t_n \gets$ temporary variables in $p_s$; \\
            $M.\Insert{$sp_e$}$; \ \ \Return $|\red{?c[p_s]}\ (s, t_1, \dots, t_n)|$;
        }
        $sp_e' \gets [\EE{$sp_e$}\ |\ sp_e \in \Child{$p_s, \mathbb E$}]$; \\
        $sp_s' \gets [\ES{$sp_s$}\ |\ sp_e \in \Child{$p_s, \mathbb S$}]$; \\
        \Return \Rewrite{$p_s, sp_e' \cat sp_s'$}.
    }
    $\phi'_s \gets \ES{$\phi_s$}$;\ \ \Return $Q, M, \lambda s. \phi'_s$;
\end{algorithm}

First, Both the objective function $o$ and keywords $\red{?R}$ are queries. In $\textit{prog}_1'$, $o$ and $\red{?R}$ are replaced with $\red{?q[o]}$ and $R' = \{(op, \red{?q[k]})\ |\ (op, k) \in R\}$. Second, constructors and other queries are extracted from $\phi$. In $\textit{prog}_1'$, $\phi$ is replaced with function $\phi'$ constructed by Algorithm \ref{alg:extract}, a structural recursion on the AST of $\phi$. The notations used in Algorithm \ref{alg:extract} are explained blow. 
\begin{itemize}
\item $|s|$ represents to the input variable of $\phi$.
\item \EE and \ES corresponds to non-terminal $\mathbb E$ and $\mathbb S$ in $\mathcal L_H$ respectively. Specially, $\EE$ returns $\bot$ if the current expression in $\textit{prog}_1$ is inside a query.
\item $\Child{$p, \mathbb N$}$ returns all children of AST node $p$ corresponding to non-terminal $\mathbb N$, and $\Rewrite{$p, c$}$ constructs a new AST node by replaces the children of $p$ with list $c$.
\end{itemize}

To characterize the specification of $\red{?f_p}, \red{?q}$ and $\red{?c}$, we introduce two notations $\F[p]$ and $RE(p, i)$.
\begin{itemize}
\item For each query (or constructor) $p$, functor $\F[p]$ indicates partial solutions in the input of $p$. For queries extracted from $\red{R?}$ and $o$, $\F[p] \coloneqq \mathsf I$; For functions extracted from $\phi$, $\F[p] \coloneqq \F \times !T_1 \times \dots !T_n$, where $\F$ is the functor used by the recursive generator in $\textit{prog}_1$, and $!T_i$ is the type of the $i$th temporary variable used by $p$\footnote{Note that there are only two ways to introduce a temporary variable in language $\mathcal L_H$, which are lambda expressions ($\lambda x. \mathbb E$) and for-loops ($\text{foreach }x \in [\mathbb E, \mathbb E] \text{ in }\mathbb S$).}. 
\item Given instance $i$, for each query (or constructor) $p$, set $RE(p, i)$ records all the inputs on which $p$ is invoked during $(\textit{prog}_1\ i)$, which can be obtained by instrumenting $\textit{prog}_1$.
\end{itemize} 

Lemma \ref{theorem:step2} provides a sufficient condition for $\red{?f_p}, \red{?q}$ and $\red{?c}$ to guarantee the correctness of $\textit{prog}_1'$. 

\begin{lemma}\label{theorem:step2} Given instance $i$ and program $\textit{prog}_1$ in Form \ref{form:thinning}, let $\textit{prog}_1'$ be result of Rule 1. If for any query $q$ and constructor $m$, Formula \ref{formula:2q} and Formula \ref{formula:2m} are satisfied respectively, $\textit{prog}_1 \sim_i \textit{prog}_1'$ holds.
\begin{gather}
    \forall e \in RE(q, i), q\ e = \red{?q[q]}\ (\mathsf F[q]\red{?f_p}\ e)  \label{formula:2q}\\
    \forall e \in RE(m, i), \red{?f_p}\ (m\ e) = \red{?c[m]}\ (\F[m]\red{?f_p}\ e)   \label{formula:2m}
\end{gather}
\end{lemma}

\smallskip
To synthesize from Formula \ref{formula:2q} and \ref{formula:2m}, \plusname first rewrite Formula \ref{formula:2q} as the follows.
\begin{gather}
(q \circ id)\ e = (\red{?q[q]} \circ \F[q](\textit{null} \triangle \red{?f_p}))\ e \label{formula:si-query}
\end{gather}
where $id$ is the identity function, i.e., $\lambda x. x$, and $\textit{null}$ is a dummy function that outputs nothing.

Compared to Definition \ref{def:lifting}, Formula \ref{formula:si-query} is equal to the specification of partial lifting problem $\mathsf {PLP}(\{id\}, q, \textit{null}, \cup RE(q, i))$. Therefore, \plusname invokes \textit{AutoLifter} to solve this task and find a solution $\red{?q[q]} = q[q]$ and $\red{?f_p} = f_p[q]$ for each query $q$.

Then, \plusname merges these results by fixing $\red{?f_p}$ to $f_q \triangle \red{?f_p'}$, where $f_q = \triangle_{q \in Q}f_p[q]$ and $Q$ represents the set of queries. For each query $q$, such a $\red{?f_p}$ and $q[q]$ can be converted into a solution to Formula \ref{formula:2q} by adjusting the input type of $q[q]$. 

Next, \plusname substitutes $\red{?f_p}$ into the Formula \ref{formula:2m} and rewrites the result into the follows. 
\begin{gather}
\!\!\!\forall m \in M, ((f_q \triangle \red{?f_p'}) \circ m)\ e = (\red{?c[m]} \circ \F[m](f_q \triangle \red{?f_p}))\ e 
\label{formula:si-mod}
\end{gather}
where $M$ represents the set of constructors. Compared to Definition \ref{def:lifting} gain, Formula \ref{formula:si-mod} is equal to the specification of lifting problem $\mathsf {LP}(M, f_p, f_p, E)$, where $E[m]$ is equal to $\cup RE(m, i)$. Therefore, the remaining part $\red{?f_p'}$ in $\red{?f_p}$ and all $\red{?c[m]}$ can be synthesized by invoking \textit{AutoLifter} again.

At last, \mainname fills the synthesized $\red{?f_p}, \red{?q[q]}$ and $\red{?c[m]}$ into the intermediate program $\textit{prog}_1'$ and thus obtains the result of applying Rule 1.

\subsection{Rule 2: Optimizing the Number of Keys}
Given program $(r, o)$, where $r$ is a recursive generator, and a set of instances $I$, \plusname optimizes the number of keys by synthesizing a proper key function $\red{?f_m}$ and returns $(r^{\red{?f_m}}, o)$ as the result.

Lemma \ref{theorem:step3} provides a sufficient condition for $\red{?f_m}$ to guarantee the correctness of $(r^{\red{?f_m}}, o)$.

\begin{lemma}\label{theorem:step3} Given instance $i$ and program $(r, o)$, where $r$ is a recursive generator, $(r^{\red{?f_m}}, o) \sim_i (r, o)$ if for any two states $s_1, s_2 \in (S_r\ i)$, $r\ s_1 \neq r\ s_2 \rightarrow \red{?f_m}\ s_1 \neq \red{?f_m}\ s_2$.
\end{lemma}

Given an instance $i$, a set of examples $ME(i)$ can be extracted according to Lemma \ref{theorem:step3}. Each example in $ME(i)$ is a pair of states, on which $\red{?f_m}$ must output different keys.

For efficiency, to limit the number of keys, \plusname requires $\red{?f_m}$ to return a tuple of scalar values, and thus synthesizes it in the form of $\lambda s. (\red{?key_1}\ s, \dots, \red{?key_n}\ s)$. At this time, the number of keys returned by $\red{?f_m}$ is bounded by the product of the ranges of all $\red{?key_i}$. Therefore, \plusname minimizes the following objective function while synthesizing $\red{?f_m}$.
$$
\textit{cost}(\red{?f_m}, I) \coloneqq \prod_{i=1}^n \textit{range}(k_i, S), \text{ where } S = \cup_{i \in I} (S_r\ i)
$$

Notice that $\red{?f_m}$ can be regraded as a preorder $\red{?R_m}$ where $s_1\red{?R_m}s_2$ is defined as $\wedge_i (\red{key_i}\ a) = (\red{key_i}\ b)$, and $\textit{cost}(\red{?f_m}, I)$ is exactly the same as $N_{\red{?R_m}}(\cup_{i \in I}, S_r\ i)$ defined in Theorem \ref{theorem:thinning}. Therefore, the synthesis task of $\red{?f_m}$ has the same form as the synthesis task for thinning discussed in Section \ref{section:main-task}, and can be solved by the synthesis algorithm proposed in Section \ref{section:main-algorithm}.

\subsection{Rule 3: Optimizing the Representation of Search States}
After applying the second rule, the program is transformed into the following form.
$$
\textit{prog}_3 = (r, o), \quad  r = rg\big(\textit{thin}[R] \circ \textit{cup} \circ \mathsf P\phi, \psi\big)_{\F}^{f_m}
$$ 
where the bodies of $\phi$ and $\psi$ are statements ($\mathbb S$) in $\mathcal L_H$.

Similar to Rule 1, \plusname optimizes the size of states via a converting function $\red{?f_s}$, which maps states in $\textit{prog}_3$ into a tuple of scalar values. \plusname rewrites all state-related queries and constructors in $\textit{prog}_3$ with functions $\red{?q}$ and $\red{?c}$ respectively, and thus constructs an intermediate program $\textit{prog}_{3}'$ that performs almost the same as $\textit{prog}_3$ except state $s$ in $\textit{prog}_3$ is stored as $\red{?f_s}\ s$. 

In $\textit{prog}_3$, key function $f_m$ is identified as a query. Other queries and constructors are extracted from $\psi$ similarly to Algorithm \ref{alg:extract}. The only difference is that in $\psi$, the parameter $sp_e$ of $|\textit{collect}\ sp_e|$ is a structure including states, on which $r$ will be recursively applied, and scalar values, which will be directly passed to $\phi$. Because $\F$ is a polynomial functor, these components can be extracted from $sp_e$ via the access operator, i.e., $|sp_e.i_1.i_2\dots i_n|$. Those components corresponding to states and scalar values are identified as constructors and queries respectively. 

We omit the concrete synthesis task and the synthesis algorithm here because they are the same as the counterparts for Rule 1.

\subsection{Properties of \plusname}
We end this section with two noticeable properties of \plusname. First, Theorem \ref{theorem:correct} guarantees the result of \mainname to be correct on all given instances.

\begin{theorem} \label{theorem:correct} Given input program $(h, o)$ where $h$ is a relational hylomorphism and a set of instances $I$, let $p^*$ be the program generated by \plusname with $I$. Then $\forall i \in I$, $(h, o) \sim_i p^*$.
\end{theorem}

Second, Theorem \ref{theorem:efficiency} shows when $\phi, \psi$ and $o$ in the input program and all functions involved in the program space run in pseudo-polynomial time and involve only linear arithmetic operators, the program synthesized by \plusname is guaranteed to run in pseudo-polynomial time, i.e., the time complexity of the result is polynomial to the size and values in the input.

\begin{theorem} \label{theorem:efficiency} Given input $(\hylo{\phi, \psi}_{\F}, o)$ and grammar $G$ specifying the program space for synthesis tasks, the program generated by \plusname must be pseudo-polynomial time if the following conditions are satisfied: (1) $\phi$, $\psi$ and programs in $G$ runs in pseudo-polynomial time, (2) each value and the size of each recursive data structure generated by the input program are pseudo-polynomial, (3) all operators in $G$ are linear, i.e., their outputs are bounded by a linear expression with respect to the input. 
\end{theorem}
\section{Discussion} \label{section:discussion}
In this section, we discuss three subtle details on the design of \mainname and \plusname. 

\noindent \textbf{The order of transformations}. \plusname optimizes the input program by applying thinning, Rule 1, Rule 2, and Rule 3 in order. Such an order is decided because of the following reasons.

First, Rule 1 should be applied after thinning. Because Rule 1 ignores most information in the solution, some attributes may become incalculable after applying Rule 1. In the example discussed in Section \ref{section:overview}, after replacing solution $p$ with $(\textit{sumw}\ p, \textit{sumv}\ p)$, the size $|p|$ becomes incalculable in the optimized program. Therefore, if Rule 1 is applied before thinning, some efficient preorders for thinning may become incalculable and thus an inefficient preorder may be used.

Second, Rule 3 should be applied after Rule 2. Before applying Rule 2, the key function used for memoization is $\lambda s. s$, which means that the whole state $s$ is necessary. Therefore, applying Rule 3 before Rule 2 cannot lead to any optimization. 

Third, thinning and Rule 1 should be applied before Rule 2. Because both thinning and Rule 1 simplify the output of the generator, applying them before Rule 2 may let the generator output the same on more states and thus may enable a more efficient key function.

\noindent \textbf{The requirement on the relational hylomorphism.} Following the thinning theorem (Theorem \ref{theorem:thinning-theorem}), \mainname requires the generator in the input program to be specified in the form of relational hylomorphism. We believe this requirement is not a significant limitation in practice because of the following two reasons.

First, as shown in Section \ref{section:moti-main}, a generator can be converted to hylomorphism by specifying the recursion of states and the construction of solutions separately. Such a conversion does not take algorithmic effort and should be easier for the user than proposing an efficient algorithm.

Second, there have been studies on automatically generating a hylomorphism from a recursive program~\cite{DBLP:conf/icfp/HuIT96}. Therefore, this limitation can be eliminated by combining \mainname and \plusname with these approaches. 

\noindent \textbf{The correctness guarantee of the result.} Both \mainname and \plusname ensure the correctness of the result only on a given set of instances. Such a guarantee can be improved to the  correctness on all instances via a complete verifier and the CEGIS framework~\cite{DBLP:conf/asplos/Solar-LezamaTBSS06}. 

However, due to the complexity of memoization, such a verifier may not exist. Therefore, in our implementation we use the probabilistic verifier provided by \textit{AutoLifter}, which verifies quickly by testing the result on a dynamically adjusted number of random instances. The guarantee provided by this verifier follows the framework of PAC learnability~\cite{DBLP:journals/cacm/Valiant84}, which ensures the probability for the error rate to be larger than a threshold is small. The practical performance of such guarantees has been demonstrated by \citet{DBLP:journals/pacmpl/JiXXH21, DBLP:conf/sigsoft/WangBKS21}.
\section{Implementation} \label{section:implementation}
Our implementation of \mainname and \plusname can be found in the supplementary material.

\noindent \textbf{Generating instances}. \mainname and \plusname require a set of instances to extract examples for each synthesis task. We generate these instances by sampling. Given an instance space specifying by assigning each integer with a range and each recursive data structure with an upperbound on the size, \mainname and \plusname sample from the space according to a uniform distribution. 

To limit the time cost of executing the input program, \mainname and \plusname decide this space by iteratively squeezing from a default space until the average time cost of the input program on a random instance is smaller than $10^{-3}$ second.

\noindent \textbf{Grammar}. \mainname uses a grammar to describe the space of possible synthesis results. In our implementation, we extended the CLIA grammar in SyGuS-Comp~\cite{padhisygus} with the following operators related to lists and binary trees.
\begin{itemize}[leftmargin=*]
\item Accumulate operator \textit{fold} and lambda expressions.
\item Access operators \textit{access} for lists, and \textit{value}, \textit{lchild}, \textit{rchild}, \textit{isleaf} for binary trees.
\item Match operator \textit{match} for lists and binary trees, which returns the first occurrence of a sublist (subtree) on a given list (tree). This operator is useful in applying Rule 3, as it can correspond the search state to some global data structure.
\end{itemize}

\noindent \textbf{Others}. The original implementation of \textit{AutoLifter} uses \textit{PolyGen}~\cite{DBLP:journals/pacmpl/JiXXH21}, a specialized solver for conditional linear expressions, to synthesize $\red{?c}$ for lifting problems. Therefore, when the input program uses non-linear operators, we will replace \textit{PolyGen} with a SOTA enumerative solver, \textit{observational equivalence}~\cite{DBLP:conf/pldi/UdupaRDMMA13}, to make \textit{AutoLifter} applicable.

To implement the preorder synthesizer introduced in Section \ref{section:main-algorithm}, we set $n_c^*$ to $2$ and $s_c^*$ to $10^{5}$, which are enough for most known tasks, and set $n_t$ to $8$, which ensures that the error rate of $\CC$ to be at most $e^{-1} \approx 37\%$.

To use the iterative verifier discussed in Section \ref{section:discussion}, we set the basic number of examples to $10^4$ and thus the probability for the error rate of the result to be more than $10^{-3}$ is at most $1.82 \times 10^{-4}$.
\section{Evaluation} \label{section:evaluation}
Our evaluation answers two research questions.
\begin{itemize}
\item \textbf{RQ1}: How is the overall performance of \mainname and \plusname?
\item \textbf{RQ2}: How do thinning and the three rules proposed in this paper perform in \plusname?
\end{itemize}

\subsection{Dataset}
Our evaluation is conducted on a dataset including $37$ input programs for $17$ different COPs. Besides problem \textit{0/1 knapsack} discussed in Section \ref{section:overview}, the other COPs are collected from \textit{Introduction to Algorithms}~\cite{cormen2009introduction}, a widely-used textbook for algorithm courses. This book introduces dynamic programming in its 15th chapter with $4$ example COPs and provides $12$ other COPs as the exercise. We include all these $16$ COPs in our dataset.

The input program is written in our language $\mathcal L_H$ defined in Section \ref{subsection:program}, where each program comprises a generator in the form of relational hylomorphism and a scorer. We divide the $17$ COPs into two categories according to whether solutions are constrained, and construct input programs for them via two criteria respectively.

The first category includes COPs where solutions are not constrained. A representative COP here is \textit{rod cutting}, the first example in the 15th chapter of \textit{Introduction to Algorithms}. 

\smallskip
\parbox{0.9\columnwidth}{\begin{center}\it
    Known that a rod of length $i$ worth $w_i$, the task is to cut a rod of length $n$ into several pieces and maximize the total value of all pieces.
\end{center}}
\smallskip

In this task, all possible ways to cut the rod are valid solutions. For each COP in this category (7 in total), we construct an input program as natural, where the generator returns all possible solutions and the scorer calculates the objective value.

The second category includes COPs where solutions are constrained. A representative COP in this category is \textit{0/1 knapsack}, where a solution is valid only when the total weight is no more than the capacity. For each COP in this category (10 in total), we implement three input programs $(g = \hylo{\phi, \psi}_{\F}, o)$ according to three extreme principles. 
\begin{itemize}
\item The first program keeps all information in the input of $\psi$ and filters out invalid solutions while deciding transitions. Figure \ref{fig:p2} shows such a program for \textit{0/1 knapsack}.
\item The second program tries all possibilities of constructing solutions in $\psi$ and lets $\phi$ filter out invalid ones. Figure \ref{fig:p1} shows such a program for \textit{0/1 knapsack}.
\item The last program uses $g$ to generate all solutions and excludes invalid ones via a small enough objective value. Figure \ref{fig:p3} shows such a program for \textit{0/1 knapsack}.
\end{itemize}

\subsection{Experiment Setup}
We run \mainname and \plusname on all $37$ tasks in the dataset and set a time limit of $120$ seconds for \mainname and each rule in \plusname. Especially, if a rule in \plusname times out, \plusname will skip this rule and go on to the next. 

For each execution, we record the time cost and the transformation result for \mainname and each rule in \plusname. We manually verify the correctness and the time complexity of each result and compare them with the reference algorithms provided by \citet{cormen2009introduction} and \citet{algsolutions}.

\subsection{RQ1: Overall Performance of \mainname and \plusname} \label{section:rq1}
\newcommand{\tc}[1]{\tilde O({#1})}
\begin{table*}[t]
    \caption{The performance of \mainname and \plusname on all tasks in our dataset.}
    \renewcommand\arraystretch{1.1}
    \begin{spacing}{1}
        \small

        \begin{threeparttable}
        \begin{tabular}{|c|c|c|c|c|c|c!{\vrule width 1.5pt}c|c|c|c|c|c|c|}
            \Xhline{1pt}
            \hline
            \multirow{2}{*}{COP} & \multirow{2}{*}{Imp} & \multirow{2}{*}{$T_{\text{inp}}$} & \multicolumn{2}{c|}{\mainname} & \multicolumn{2}{c!{\vrule width 1.5pt}}{\plusname}& \multirow{2}{*}{COP} & \multirow{2}{*}{Imp} & \multirow{2}{*}{$T_{\text{inp}}$} & \multicolumn{2}{c|}{\mainname} & \multicolumn{2}{c|}{\plusname}\\
            \cline{4-7} \cline{11-14}
            & & & $T_{\text{res}}$ & $\text{Time}$ & $T_{\text{res}}$ & $\text{Time}$ & & & & $T_{\text{res}}$ & $\text{Time}$ & $T_{\text{res}}$ & $\text{Time}$ \\
            \Xhline{1pt}
            \multirow{2}{*}{0/1} & 1 & \multirow{3}{*}{$2^n$} & \multirow{1}{*}{$2^n$} & $0.1$ & \multirow{3}{*}{$n^2\dag$}& $25.9$& \multirow{3}{*}{15-6} & 1 & \multirow{3}{*}{$2^n$} & $2^n$ & $0.5$ & \multirow{3}{*}{$n\dag$} & $41.2$\\
            \cline{2-2} \cline{4-5} \cline{7-7} \cline{9-9} \cline{11-12} \cline{14-14}
            \multirow{2}{*}{Knapsack}& 2 & &\multirow{2}{*}{$n^3$} & $1.8$ & &  $11.9$& & 2 & & \multirow{2}{*}{$n^2$} & $0.4$ & & $28.1$\\
            \cline{2-2} \cline{5-5} \cline{7-7} \cline{9-9} \cline{12-12} \cline{14-14}
            & 3 & & & $3.8$ & & $23.8$& & 3& & & $0.8$ & & $31.4$ \\
            \hline
            \multicolumn{2}{|c|}{Rod Cutting} & $2^n$& $n^3$ & $0.1$ & $n^{2}\dag$ & $18.4$& \multirow{3}{*}{15-7} & 1& \multirow{3}{*}{$n^n$} & \multirow{3}{*}{$n^3\dag$} & $0.1$& \multirow{3}{*}{$n^3\dag$}&$3.6$\\
            \cline{1-7} \cline{9-9} \cline{12-12} \cline{14-14}
            \multicolumn{2}{|c|}{Matrix Chain} & $4^n$& $n^4$& $0.1$&$n^{3}\dag$& $15.4$&  & 2 & & & $0.1$ & & $10.5$\\
            \cline{1-7} \cline{9-9} \cline{12-12} \cline{14-14}
            \multirow{3}{*}{LCS} & 1 & \multirow{3}{*}{$5.8^n$} & \multirow{3}{*}{$n^3$}& $0.1$& \multirow{3}{*}{$n^{2}\dag$}& $6.9$& & 3 & & & $0.2$ & & $8.9$\\
            \cline{2-2} \cline{5-5} \cline{7-14}
            & 2& & & $0.1$& &$5.9$ & \multirow{3}{*}{15-8} & 1 & $2^n$ & $2^n$ & $0.1$ & $n^2\dag$ & $31.3$\\
            \cline{2-2} \cline{5-5} \cline{7-7} \cline{9-14}
            &3& & & $0.3$ & & $24.3$ & & 2 & \multirow{2}{*}{$n^n$} & \multirow{2}{*}{$n^4$} & $0.2$ & \multirow{2}{*}{$n^3$} & $25.5$\\
            \cline{1-7} \cline{9-9} \cline{12-12} \cline{14-14}
            \multicolumn{2}{|c|}{Optimal BST} & $4^n$ &$n^4$& $0.2$ & $n^3\dag$& $24.5$& & 3 & & &$0.2$ & & $12.1$
            \\
            \cline{1-14} 
            \multicolumn{2}{|c|}{15-1} & $2^n$& $n^4$ & $0.1$ & $n^{4}$ & $123.7$& \multicolumn{2}{c|}{15-9} & $4^n$& $n^4$& $0.1$ & $n^3\dag$ & $4.8$\\
            \cline{1-3} \cline{4-14}
            \multirow{3}{*}{15-2} & 1 & \multirow{3}{*}{$2^n$}& \multirow{3}{*}{$n^3$} & $0.1$ & \multirow{3}{*}{$n^2\dag$} &$8.7$ & \multicolumn{2}{c|}{15-10} & $n^n$& $n^4$ & $0.1$ & $n^4$ & $34.1$\\
            \cline{2-2} \cline{5-5} \cline{7-7} \cline{8-14}
            & 2& & & $0.1$& & $5.6$ & \multirow{3}{*}{15-11} & 1 & \multirow{3}{*}{$n^n$} & $n^n$ & $0.1$ & $n^3\dag$ & $142.0$\\
            \cline{2-2} \cline{5-5} \cline{7-7} \cline{9-9} \cline{11-14}
            &3& & & $0.3$& & $20.9$ & & 2 & & \multirow{2}{*}{$n^4$}& $9.0$ & \multirow{2}{*}{$n^4$}& $137.5$\\
            \cline{1-7} \cline{9-9} \cline{12-12} \cline{14-14}
            \multicolumn{2}{|c|}{15-3} & $2^n$ & $n^3$ & $1.5$ & $n^3$ & $133.4$ & & 3 & & & $9.9$ & & $137.1$\\
            \cline{1-14}
            \multirow{3}{*}{15-4}& 1 & \multirow{3}{*}{$2^n$}& \multirow{3}{*}{$n^3$} & $0.1$ & \multirow{3}{*}{$n^3$}& $141.4$ & \multirow{3}{*}{15-12} & 1 & \multirow{3}{*}{$n^n$} & $n^n$ & $0.1$ & \multirow{3}{*}{$n^3\dag$} & $54.9$\\
            \cline{2-2} \cline{5-5} \cline{7-7} \cline{9-9} \cline{11-12} \cline{14-14}
            & 2& & & $0.2$ & & $240.5$ & & 2 & & \multirow{2}{*}{$n^5$} & $2.1$ & & $27.6$\\
            \cline{2-2} \cline{5-5} \cline{7-7} \cline{9-9} \cline{12-12} \cline{14-14}
            & 3& & & $0.3$ & & $240.6$ & & 3 & & & $2.9$ & & $31.6$\\
            \hline \Xcline{8-14}{1pt}
            \multirow{3}{*}{15-5} & 1 & \multirow{2}{*}{$7.6^n$}& \multirow{2}{*}{$n^3$} & $0.1$& \multirow{2}{*}{$n^2\dag$}& $40.1$\\
            \cline{2-2} \cline{5-5} \cline{7-7}
            & 2 & & & $0.1$& & $41.3$\\
            \cline{2-7}
            & 3 & \multicolumn{2}{c|}{$ \approx7.7^n\ddag$} & $120.0$& $7.7^n$ &  $274.9$\\
            \Xcline{1-7}{1pt}

        \end{tabular}
        \begin{tablenotes}[para]
        \small
          \item[\dag] The result achieve the same time complexity as the reference algorithm for the corresponding COP. \\ 
          \item[\ddag] Both time complexities are $\tilde \Theta(a^n)$, where $\sqrt a \approx 2.77$ is the largest real root of $x^4 - 2x^3 -2x^2 - 1$.
        \end{tablenotes}
    \end{threeparttable}
    \label{table:rq1}
    \end{spacing}
\end{table*}

We confirm that all programs generated by \mainname and \plusname are completely correct. The detailed performance of \mainname and \plusname are listed in Table \ref{table:rq1}.
\begin{itemize}
\item For each task, \textit{COP} lists the name of the corresponding COP in \textit{Introduction to Algorithm}, \textit{Imp} lists the index of the principle used to implement the input program if the COP is in the second category, and $T_{\textit{inp}}$ lists the time complexity of the input program in $\tilde \Omega$ notation.
\item For each approach, $T_{\textit{res}}$ lists the time complexity of the result program in $\tilde O$ notation, and \textit{Time} lists the number of seconds used to generate the result.
\end{itemize}

First, Table \ref{table:rq1} demonstrates the effectiveness of \mainname on automating thinning. \mainname achieves exponential speed-ups against the input program on $31/37(83.8\%)$ tasks with an average time cost of $4.2$ seconds. Besides, the comparison between the results of \mainname and the reference algorithms demonstrates the gap between thinning and human experts: \mainname achieves the same time complexity as the reference algorithm only on $3/37 (8.1\%)$ tasks.

Second, Table \ref{table:rq1} demonstrates the overall effectiveness of \plusname on synthesizing efficient memoization algorithms. \plusname achieves exponential speed-ups against the input program on $36/37(97.3\%)$ tasks with an average time cost of $59.2$ seconds. Moreover, compared to \mainname, the ability of \plusname is much closer to human experts: \plusname achieves the same time complexity as the reference algorithm on $26/37 (70.3\%)$ tasks. 

Besides, we conduct a case study on those 11 tasks where \plusname fails in achieving the same complexity as the reference algorithms, and conclude the following three main reasons.
\begin{itemize}
\item For task (15-5, 3), \plusname fails because the scorer in this task provides little information. In COP 15-5, a solution is a sequence of editions, and a solution is valid if the results of these editions are equal to the target. In this sense, almost all partial solutions are invalid, on which the scorer in (15-5, 3) simply returns $-\infty$ according to the third principle. At this time, \plusname can hardly extract examples for $\red{?R}$ and thus fails in applying thinning. 
\item For tasks (15-8, 2) and (15-8, 3), \plusname fails because of useless transitions. In both tasks, $\psi$ generates $O(n)$ transitions but only $O(1)$ among them can lead to valid solutions. However, as all the three supplementary rules in \plusname keep the recursive structure unchanged, this non-optimal behavior remains in the result and thus leads to a higher time complexity. Optimizing the recursive structure of hylomorphism is future work.
\item For the other $8$ tasks, \plusname fails because their input programs involve complex non-linear expressions. As mentioned in Section \ref{section:implementation}, \plusname uses an enumerative solver, namely \textit{observational equivalence}, to synthesize $\red{?c}$ required by the first and the third supplementary rules when non-linear expressions are involved. Because the scalability of this solver is limited, \plusname times out while applying the first and the third rules when the target $\red{?c}$ is non-linear and too large, e.g., $x_1 - \text{dis}(p_1, p_2) + \text{dis}(p_1, p_3) + \text{dis}(p_2, p_3)$ in Task 15-3, where $\text{dis}(p,q)$ is the abbreviation of $(p.x-q.x)^2 + (p.y-q.y)^2$. This fact shows that the effectiveness of \mainname can be further improved by designing efficient solvers for synthesizing $\red{?c}$.
\end{itemize}

\subsection{RQ2: Performance of Thinning and Supplementary Rules in \plusname}
\begin{table}[t]
    \caption{The performance of each transformation step.}
    \renewcommand\arraystretch{1.1}
    \begin{spacing}{1}
        \small
        \begin{tabular}{|c|c|c|c|c!{\vrule width 1.5pt}c|c|c|c|c|}
            \Xhline{1pt}
            \hline
            Name& Exp & Poly & $\bot$ & Time & Name& Exp & Poly & $\bot$ & Time \\
            \Xhline{1pt}
            \hline
            Thinning & $31$& $0$&$1$ & $4.2$  &Rule 1 & $0$& $3$& $7$&$32.8$\\
            \hline
            Rule 2 & $5$& $0$ & $0$& $5.1$ & Rule 3& $0$& $24$& $3$&$59.2$\\
            \hline
            \Xhline{1pt}
        \end{tabular}
    \label{table:rq2}
    \end{spacing}
\end{table}
We manually analyze all intermediate results of \plusname and summarize the performance of thinning and each supplementary rule in Table \ref{table:rq2}, where \textit{Name} lists the name of the corresponding rule, \textit{Exp} lists the number of tasks where the rule achieves an exponential speed-up, \textit{Poly} lists the number of tasks where the rule achieves a polynomial speed-up, $\bot$ lists the number of failed tasks, and \textit{Time} records the average time cost (seconds) on all tasks.

According to Table \ref{table:rq2}, thinning and the second rule produce all exponential speed-ups, while the first rule and the third rule produce all polynomial ones. This result matches the target of each rule: (1) the number of solutions, optimized by thinning, and the number states, optimized by the second rule, can be exponential, and (2) the scale of solutions, optimized by the first rule, and the scale of the states, optimized by the third rule, are usually polynomial.

Note that the effects of the first two rules seem to be insignificant in Table \ref{table:rq2} because they are usually \textit{delayed} by the third rule. In many cases, the time cost of operating states forms a bottleneck of the time complexity and thus the effects of the first two rules will not be revealed until the third rule is applied. For instance, in the example discussed in Section \ref{section:overview}, 
each of the three rules reduces an $O(n)$ component to $O(1)$, but the overall time complexity would not change if any of the components remains.
 
\section{Related Work} \label{section:related}
\noindent \textbf{Program Calculation}. This paper is related to those studies for deriving dynamic programming in program calculation. First, \citet{DBLP:conf/plilp/Moor95, DBLP:conf/flops/MorihataKO14,DBLP:conf/pepm/Mu08,DBLP:books/daglib/0096998} manually derive dynamic programming algorithms by \textit{thinning}. Among them, \citet{DBLP:conf/flops/MorihataKO14} notice that applying thinning solely may not be enough to derive an efficient dynamic programming algorithm, and uses a rule namely \textit{incrementalization} to optimize the program generated by thinning. Compared to the supplementary rules proposed in our paper, the application of this rule is still manual and is restricted to associative and commutative operators.

Second, there are several existing studies on automating thinning. \citet{DBLP:journals/jfp/Bird01a, sasano2000make} focus only on a special kind of COPs namely \textit{maximum marking problem}, where a partial solution is a set of weighted items and the objective function is the total weight of the selected items. \citet{DBLP:journals/ngc/Morihata11} focuses on a special kind of hylomorphisms namely sequential decision procedures, which recurses strictly according to the structure of a list. In comparison, the scopes of these three approaches are strictly more narrow than our approaches \mainname and \plusname. Only $9/37 (24.3\%)$ tasks in our dataset are instances of the maximum marking problem, and only $10/37 (27.0\%)$ tasks in our dataset use sequential decision procedures.

Last, there are other approaches for deriving dynamic programming beside thinning~\cite{DBLP:conf/oopsla/PuBS11, DBLP:journals/scp/GiegerichMS04, DBLP:journals/csur/PettorossiP96, DBLP:conf/sigsoft/0003ML21, DBLP:conf/ppdp/SauthoffJG11, DBLP:journals/lisp/LiuS03}. Most of them are manual or semi-automated. The only automated approach we know is a transformation rule that automatically generates dynamic programming for a sequential decision procedure on lists when several conditions are satisfied~\cite{DBLP:conf/sigsoft/0003ML21}, of which the scope is strictly more narrow than ours due to the requirement on sequential decision procedures.

\noindent \textbf{Program Synthesis}. There have been many synthesizers proposed for automatically synthesizing algorithms or efficient programs~\cite{DBLP:conf/pldi/SmithA16,DBLP:conf/pldi/MoritaMMHT07,toronto21,DBLP:conf/pldi/FedyukovichAB17,acar2005self,DBLP:conf/pldi/KnothWP019,DBLP:conf/cav/HuCDR21}, but none of them are for synthesizing dynamic programming algorithms.

Our approaches use program synthesizers to (1) synthesize preorders for thinning and Rule 2, and (2) synthesize program fragments for Rule 1 and Rule 3. The specification of both tasks are instances of relational specification, and thus our approach is related to \textit{Relish}~\cite{DBLP:journals/pacmpl/0001WD18}, a general solver for relational specifications. However, \textit{Relish} cannot be applied to our tasks because (1) \textit{Relish} cannot optimize an objective function while synthesis, and (2) the \textit{finite tree automata} used by \textit{Relish} does not directly support lambda expressions, which is included in our grammar.

There are also many synthesizers for input-output specifications~\cite{DBLP:journals/pacmpl/JiS0H20, DBLP:conf/popl/Gulwani11, DBLP:conf/pldi/FeserCD15, DBLP:conf/pldi/OseraZ15, DBLP:conf/iclr/BalogGBNT17, DBLP:journals/pacmpl/WangDS18}. However, such specifications do not exist in both synthesis tasks, and thus all these synthesizers are unavailable.

\section{Conclusion}\label{section:conclusion}
In this paper, we propose two novel synthesizers \mainname and \plusname for automatically synthesizing efficient memoization for relational hylomorphism. We demonstrate the efficiency and effectiveness of both approaches on a dataset including $37$ tasks related to $17$ COPs in our evaluation. 

This paper is motivated by the theories in program calculation for deriving dynamic programming algorithms. We notice that there are also studies for deriving other algorithms such as \textit{greedy algorithm}~\cite{DBLP:conf/ifip2-1/BirdM93, helman1989theory} and \textit{branch-and-bound}~\cite{DBLP:journals/scp/Fokkinga91} from relational hylomorphisms. Extending our approaches to support these algorithms is future work.
\bibliography{ref}
\clearpage
\appendix
\section{Appendix}
In this section, we complete the proofs of the lemmas and the theorems in our paper.

\begin{theorem} [Theorem \ref{theorem:thinning}] Given a keyword preorder $R$ of $\{(op_i, k_i)\}$, define function $N_R(S)$ as the following, where $\textit{range}(k, S)$ is the range of $k$ on $S$, i.e., $\max_{a \in S}(k\ a) - \min_{a \in S}(k \ a)$, and $\max_1(S)$ returns the largest element in $S$ with default value $1$.
    $$
    N_R(S) \coloneqq \left( \prod_{i}\textit{range}(k_i, S)\right) \bigg/ \max_i\big(\textit{range}(k_i, S)\ \big|\ op_i \in \{\leq, \geq\}\big)  
    $$
    \begin{itemize}
    \item For any set $S$, $|\textit{thin}[R]\ S| \leq N_R(S)$.
    \item There is an implementation of $\textit{thin}[R]$ with time complexity $O(N_R(S)\textit{size}(R) + T_R(S))$, where $S$ is the input set, $\textit{size}(R)$ is the number of comparisons in $R$, $T_R(s)$ is the time complexity of evaluating all key functions in $R$ for all elements in $S$.
    \end{itemize}
    \end{theorem}
\begin{proof} Let $K_{\leq}$ be the set of key functions in $R$ with operator $\leq$, and let $k^*$ be the key function in $K_{\leq}$ with the largest range on $S$, i.e., $\arg \max \textit{range}(k, S), k \in K_{\leq}$. Especially, when $K_{\leq}$ is empty, $k^*$ is defined as the constant function $\lambda x. 0$.

    Let $K = \{k_1, \dots, k_m\}$ be the set of key functions in $R$ excluding $k^*$. According to the definition of $N_R(S)$, we have the following equality.
    $$
    N_R(S) = \prod_{i=1}^m \textit{range}(k_i, S)
    $$

We start with the first claim. Define feature function $f_k$ as $k_1 \triangle \dots \triangle k_m$. Then $N_R(S)$ is the range of $f_k$. By the definition of the keyword preorder, we have the following formula:
$$
f_k\ a = f_k\ b \rightarrow (aRb \leftrightarrow k^*\ a \leq k^*\ b)
$$
In other words, for elements where the outputs of the feature function are the same, their order in $R$ is total. Therefore, the number of maximal values in $S$ is no more than the range of the key function, i.e., $N_R(S)$.

Then, for the second claim, Algorithm \ref{appendix-alg:thin} shows an implementation of $\textit{thin}$. The time complexity of the first loop (Lines 6-10) is $O(T_R(S))$ and the time complexity of the second loop (Lines 11-20) is $O(N_R(S)\textit{size}(R))$. Therefore, the overall time complexity of Algorithm \ref{appendix-alg:thin} is $O(N_R(S)\textit{size}(R) + T_R(S))$. 
\begin{algorithm}[t]
    \small
    \caption{An implementation of $\textit{thin}[R]$.}
    \KwIn{A set $S$ of elements.}
    \KwOut{A subset including all maximal values in $S$.}
    \label{appendix-alg:thin}
    \LinesNumbered
    \SetKwProg{Fn}{Function}{:}{}
    Extract $k^*$ and $f_k = k_1 \triangle \dots \triangle k_m$ from $R$; \\
    $op_i \gets $ the operator corresponding to $k_i$; \\
    $[mi_i, ma_i] \gets$ the range of $k_i$ on $S$; \\
    $\mathbb W \gets [mi_1, ma_1] \times [mi_2, ma_2] \times \dots \times [mi_m, ma_m]$; \\
    $\forall w \in \mathbb W, \textit{Val}[w] \gets \bot$; \\
    \ForEach {$ a \in S$}{
        \If {$\textit{Val}[f_k\ a] = \bot \vee k^*\ (\textit{Val}[f_k\ a]) \leq k^*\ a$}{
            $\textit{Vak}[f_k\ a] \gets a$;
        }
    }
    \ForEach {$i \in [1, m]$}{
        \lIf{$op_i \in \{=\}$}{\textbf{continue}}
        \ForEach {$w \in \mathbb W$ in the decreasing order of $w.i$}{
            \lIf{$w.i = mi_i \vee \textit{Val}[w] =\bot $}{\textbf{continue}}
            $w' \gets w$; \quad $w'.i \gets w.i - 1$; \\
            \If{$\textit{Val}[w'] = \bot \vee k^*\ \textit{Val}[w'] \leq k^*\ \textit{Val}[w]$}{
                $\textit{Val}[w'] \gets \textit{Val}[w]$; \\
            }
        }
    } 
    \Return$ \{a \mid a \in S \wedge \textit{Val}[f_k\ a] = a\}$;
\end{algorithm}

The remaining task is to prove the correctness of Algorithm \ref{appendix-alg:thin}. Let $\mathcal A_x$ be the algorithm weakened from Algorithm \ref{appendix-alg:thin} by replacing the loop upper bound in Line $11$ from $m$ to $x$. Besides, let $R_x$ be the keyword preorder $\{(op_i, id)_{i=1}^x\} \cup \{(=, id)_{i=x+1}^m\}$. Now, consider the following claim.
\begin{itemize}
\item After running $\mathcal A_x$ on set $S$, the value of $\textit{Val}[w]$ is equal to $\arg \max_{a} k^*\ a, a \in S \wedge wR_x(f_k\ a)$. If there is no such $a$ exist, $\textit{Val}[w]$ is equal to $\bot$. 
\end{itemize}

If this claim holds, after running $\mathcal A_m$, i.e., Algorithm \ref{appendix-alg:thin} on $S$, $\textit{Val}[f_k\ a] = a$ if and only if $a$ is a local maximal in $S$. Therefore, we get the correctness of Algorithm \ref{appendix-alg:thin}.

To prove this claim, we make an induction on $m$. First, when $m$ is equal to $0$, this claim holds because only the element with the largest output of $k^*$ is retained while initializing $\textit{Val}$ (Lines 6-10).

Then, for any $x \in [1,m]$, assume that the claim holds for $\mathcal A_{x-1}$. When $op_x$ is equal to $=$, the correctness of $\mathcal A_{x-1}$ directly implies the correctness of $\mathcal A_x$. Therefore, we consider only the case where $op_x \in \{\leq\}$ below.

Let $\textit{Val}'$ be the value of $\textit{Val}$ after running $\mathcal A_{x-1}$, and let $\textit{Val}$ be the value of $\textit{Val}$ after running $\mathcal A_{x}$. For any $w \in W$ and $i \in [mi_x, ma_x]$, let $w_i$ be the feature that $\forall j \neq x, w_i.j = w.j$ and $w_i.x = i$. According to Lines 11-20 in Algorithm \ref{appendix-alg:thin}, $\textit{Val}[w]$ is equal to the element with the largest output of $k^*$ among $\textit{Val}'[w_{w.i}], \textit{Val}'[w_{w.i + 1}], \dots, \textit{Val}'[w_{\textit{ma}_x}]$. (For simplicity, we define $k^*\ \bot$ as $-\infty$). 

Assume that the claim does not hold for $\mathcal A_{x}$. Then, there exists $w \in \mathbb W$ and $a \in S$ satisfying the following formula.
\begin{align*}
    &k^*\ \textit{Val}[w] < k^*\ a \wedge wR_x(f_k\ a) \\
    \implies & k^*\ \textit{Val'}\left[w_{k_x\ \!a}\right] < k^*\ a \wedge w_{k_x\ \!a}R_{x-1}(f_{k}\ a) 
 \end{align*}
This fact contradicts with the inductive hypothesis and thus the induction holds. 
\end{proof}

\begin{theorem}[Theorem \ref{theorem:thinning-theorem}] Given program $(h\!= \!\hylo{\phi, \psi}_{\F},o)$ and preorder $R$, for any instance $i$,  $(rg(\textit{thin}[R] \circ \textit{cup} \circ \mathsf P \phi, \psi)_{\F}, o) \sim_i (h, o)$ if the following two conditions are satisfied.
    \begin{enumerate}
    \item $\forall s \in S_h\ i, \forall p_1, p_2 \in h\ s, p_1Rp_2 \rightarrow (o\ p_1 \leq o\ p_2)$.
    \item $\forall s \in S_h\ i, \forall \overline{p_1} = (p_{1,1}, \dots, p_{1,k}), \overline{p_2} = (p_{2,1}, \dots, p_{2,k})$, where $p_{1,i}$ and $p_{2,i}$ are partial solutions of the same search state for all $i \in [1, k]$, the following formula is always satisfied. 
    \begin{align}
    \bigwedge_{i=1}^k p_{1,i}Rp_{2,i} \rightarrow \forall p_1', \bigg(\overline{p_1} \twoheadrightarrow_{h,s} p_1' \rightarrow \exists p_2', \big(\overline{p_2}\twoheadrightarrow_{h,s} p_2' \wedge p_1' R p_2'\big)\bigg) \label{formula:thinning2}
    \end{align}
    \end{enumerate}
\end{theorem}
\begin{proof} For simplicity, we use $r$ to denote $rg(\textit{thin}[R] \circ \textit{cup} \circ \mathsf P \phi, \psi)_{\F}$. Because $\psi$ in $h$ is also used in $r$, the search tree generated by $r$ and $h$ on instance $i$ are exactly the same. 

For simplicity, we use $S_1 \sqsupseteq_R S_2$ to denote that elements in $S_1$ dominates elements in $S_2$ in the sense of preorder $R$, i.e., $\forall a \in S_2, \exists b \in S_1, aRb$. By the definition of \textit{thin}, for any preorder $R$ and any set $S$, $\textit{thin}[R]\ S \sqsupseteq_R S$ always holds. 

Let us consider the following claim.
    \begin{itemize}
        \item For any state $s$ in $S_h\ i$, $r\ s \subseteq h\ s \wedge r\ s \sqsupseteq_R h\ s$. 
    \end{itemize}
Let $p_o$ be any solution with the largest objective value in $h\ i$. If this claim holds, there must be a solution $p^*$ in $r\ i$ such that $p_oRp^*$. By the precondition that $(\leq, o) \in R$, $o\ p^* \geq o\ p_o$. Because $p^* \in r\ i \subseteq h\ i$, we have $o\ p^* = o\ p_o$. Therefore, at least one solution with the largest objective value are retained in $r\ i$, which implies that $(r, o) \sim_i (h, o)$. 

We prove this claim by structural induction on the search tree. 
First, $r\ s \subseteq h\ s$ can be obtained by the definition of $rg$ and $\hylo{\phi, \psi}_{\F}$. Let us unfold the definition of $h$ and $r$.
\begin{align*}
    h &= \textit{cup} \circ \mathsf P \phi \circ \textit{cup} \circ \mathsf P(\textit{car}[\F] \circ \mathsf F h) \circ \psi \\
    r &= \textit{thin}[R] \circ \textit{cup} \circ \mathsf P \phi \circ \textit{cup} \circ \mathsf P(\textit{car}[\F] \circ \mathsf F r) \circ \psi \\
\end{align*}

Starting from the inductive hypothesis, we have the following derivation. 
\begin{align*}
    &\forall s' \in T_h\ s, r\ s' \subseteq h\ s' \\
    \implies& \forall t \in \psi\ s, (\textit{car}[\F] \circ \mathsf F r)\ t \subseteq (\textit{car}[\F] \circ \mathsf F h)\ t \\
    \implies& (\textit{cup} \circ \mathsf P(\textit{car}[\F] \circ \mathsf F r) \circ \psi)\ s \subseteq (\textit{cup} \circ \mathsf P(\textit{car}[\F] \circ \mathsf F h) \circ \psi)\ s \\
    \implies & r\ s \subseteq h\ s
\end{align*}
By the induction, we prove that $\forall s \in S_h\ i, r\ s \subseteq h\ s$. 

The remaining task is to prove $\forall s \in S_h\ i, r\ s \sqsupseteq_R h\ s$. For any state $s$, let be the set of partial solutions constructed in $r\ s$ before applying $\textit{thin}[R]$. Let us consider another claim. 
\begin{itemize}
\item For any state $s$ in $S_h\ i$, $P_s \sqsupseteq_R h\ s$. 
\end{itemize}
If the second claim holds, we prove the first claim by $$r\ s = \textit{thin}[R]\ P_s \sqsupseteq_R P_s \sqsupseteq_R h\ s$$

Therefore, the remaining task is to prove the second claim via the inductive hypothesis. Suppose this claim does not hold for state $s$. 
\begin{align} 
    P_s \not \sqsupseteq_R h \ s \implies \exists p \in h\ s, \forall p' \in P_s, \neg pRp' \label{formula:assump-step1}
\end{align}
Suppose partial solution $p$ is constructed from partial solutions $p_1, \dots, p_k$ where $p_i$ is taken from state $s_i$. By the inductive hypothesis, for each $i \in [1,k]$, there exists $p_{i}' \in r\ s_i$ such that $p_i \red{?R} p_{i}'$. Let $\overline{p} = (p_1, \dots, p_k)$ and $\overline{p'} = (p_1', \dots, p_k')$. 
\begin{align}
    \bigwedge _{i=1}^k p_i \red{?R} p'_i
    \implies &\forall p_1', \big(\overline{p} \twoheadrightarrow_{h,s} p_1' \rightarrow \exists p_2', \overline{p'}\twoheadrightarrow_{h,s} p_2' \wedge p_1' \red{?R} p_2'\big) \nonumber \\
    \implies& \exists p_2', \overline{p'}\twoheadrightarrow_{h,s} p_2' \wedge p \red{?R} p_2' \nonumber \\
    \implies& \exists p_2' \in P_s, p \red{?R} p_2'  \label{formula:contradict-step1}
\end{align}
Formula \ref{formula:contradict-step1} contradicts with Formula \ref{formula:assump-step1}. Therefore, we prove the second claim, and thus the induction holds.
\end{proof}

\begin{lemma}[Lemma \ref{lemma:monoce}] Given instance $i$, for any two keyword preorders $R_1, R_2$ where all comparisons in $R_1$ are included in $R_2$, the following formula is always satisfied.
    $$
    \forall (\overline{p_1}, \overline{p_2}) \in CE(R_1, i), (\overline{p_1}, \overline{p_2}) \notin CE(R_2, i) \leftrightarrow \neg \overline{p_1}(R_2/R_1) \overline{p_2}
    $$
where $R_2/R_1$ represents the keyword preorder formed by the comparisons in $R_2$ that are not used in $R_1$.
\end{lemma}
\begin{proof} We start with the $\leftarrow$ direction. Suppose there is an example $e$ in $CE(R_1, i)$ satisfying $\exists (p_1, p_2) \in e, \neg p_1(R_2/R_1)p_2$. As the comparisons in $R_2/R_1$ are included in $R_2$, this premise implies $\exists (p_1, p_2) \in e, \neg p_1R_2p_2$. By Formula \ref{formula:thinning-eq}, $e$ cannot be a counter example for $R_2$, i.e., $e \notin CE(R_2, i)$.
    
    For the $\rightarrow$ direction, suppose there is an example $e$ in $CE(R_1, i)$ such that $\forall (p_1, p_2) \in e, p_1(R_2/R_1)p_2$. 
    
    Let $(p_{1,1}, p_{2,1}), \dots, (p_{1,n}, p_{2,n})$ be all pairs in example $e$, let $\overline{p_1}$ and $\overline{p_2}$ be the sequences of $p_{1,j}$ and $p_{2,j}$ respectively. Be the definition of $CE$, we have (1) $\forall (p_1, p_2) \in e, p_1 R_1 p_2$, (2) the following formula.

    \begin{align}
    &\exists p'_1, \overline{p_1} \twoheadrightarrow_{s} p_1' \wedge \forall p'_2, \big(\overline{p_2} \twoheadrightarrow_s p'_2 \rightarrow \neg p_{1}' R_1 p_{2}' \big) \nonumber \\
    \implies& \exists p'_1, \overline{p_1} \twoheadrightarrow_{s} p_1' \wedge \forall p'_2, \big(\overline{p_2} \twoheadrightarrow_s p'_2 \rightarrow \neg p_{1}' R_2 p_{2}' \big)  \label{formula:mono1}
    \end{align}

    By the definition of keyword preorders, we have the following derivation.

    \begin{align}
        &\forall (p_1, p_2) \in e, p_1 R_1 p_2  \wedge \forall (p_1, p_2) \in e, p_1 (R_2/R_1)p_2 \nonumber\\
        \implies& \forall (p_1, p_2) \in e, \forall (op, k) \in R_2, (k\ p_1)op(k\ p_2) \nonumber \\
        \implies& \forall (p_1, p_2) \in e, p_1R2p_2 \label{formula:mono2}
    \end{align}

    Combining Formula \ref{formula:mono2} with \ref{formula:mono1}, we know example $e$ is in $CE(R_2, i)$, and the other direction of this lemma is proved.
\end{proof}

\begin{lemma}[Lemma \ref{lemma:candidatecmp}] Given a set of instances $I$, for any two keyword preorders $R_1, R_2$ where all comparisons in $R_1$ are included in $R_2$ and $\forall i \in I, CE(R_2, i) = \emptyset$, there exists a comparison $(op, k) \in R_2/R_1$ satisfying at least $1/(|R_2| - |R_1|)$ portion of examples in $CE(R_1, I) =\cup_{i \in I} CE(R_1, i)$, i.e., 
    \begin{align*}
    \left|\left\{(p_1, p_2)\in CE(R_1, I)\ \big|\ \neg \big((k\ p_1)op(k\ p_2)\big)\right\}\right| \geq  |CE(R_1, I)| \big / (|R_2|-|R_1|)
    \end{align*}
    where $|R|$ represents the number of comparisons in keyword preorder $R$.
    \end{lemma}

    \begin{proof} Let $(op_1, k_1), \dots, (op_n, k_n)$ be comparisons in $R_2/R_1$. Define keyword preorders $R^p_x$ as $R_1 \cup \{(op_j, k_j)_{j=1}^x\}$, and define $R^a_x$ as $R_1 \cup \{(op_x, k_x)\}$. By the definition of keyword preorders, this lemma is equivalent to the following formula. 
        \begin{align}
        \exists x \in [1, n], \big|CE(R_1, I) \big/ CE(R_x^a, I)\big| \geq \frac{|CE(R_1, I)|}{n} \label{formula:cmp1}
        \end{align}
    
    We prove Formula \ref{formula:cmp1} in two steps. First, we prove that $\forall x \in [1,n]$ satisfies the following formula.
    \begin{align}
    \big|\big(CE(R^p_{x}, I) / CE(R_1, I) \big) \big / \big(CE(R^p_{x-1}, I)\big / CE(R_1, I)\big) \big| \leq \big|CE(R_1, I) \big/ CE(R_x^a, I)\big| \label{formula:cmp2}
    \end{align}

    For any $x$, let $C^p_x$ be the set in the left-hand side and let $C^a_x$ be the set in the right-hand side. Then, by Lemma \ref{lemma:monoce}, 
    \begin{align*}
    e \in C^p_x \iff& \forall (p_1, p_2) \in e, \forall j \in [1, x - 1], (k_j\ p_1)op_j(k_j\ p_2) \\
    &\wedge \exists (p_1, p_2) \in e, \exists j \in [1, x], \neg (k_j\ p_1)op_j (k_j\ p_2) \\
    \implies& \exists (p_1, p_2) \in e, \neg (k_x\ p_1) op_x (k_x\ p_2) \\
    \iff& e \in C^a_x
    \end{align*}
    Therefore, $|C^p_x| \leq |C^a_x|$ and thus Formula \ref{formula:cmp2} is proved.

    Then, we prove the following formula.
    \begin{align}
        \exists x \in [1, n], |C^p_x| \geq \frac{|CE(R_1, I)|}{n} \label{formula:cmp3}
    \end{align}

    Because $CE(R_2, I) = \emptyset$, we know $C^p_1 \cup C^p_x \cup \dots \cup C^p_n = CE(R_1, I)$. Therefore, $\sum_{i=1}^n |C^p_i| = |CE(R_1, I)|$. Let $x^*$ be the index where $|C^p_x|$ is maximized.
    $$
    n\left|C^p_{x^*}\right| \geq \sum_{i=1}^n \left|C^p_i\right| = |CE(R_1, I)|
    $$
    Therefore, we prove that Formula \ref{formula:cmp3} holds for $x = x^*$.

    As the combination of Formula \ref{formula:cmp2} and Formula \ref{formula:cmp3} implies Formula \ref{formula:cmp1}, the target lemma is proved. 
    \end{proof}

\begin{theorem} [Theorem \ref{theorem:complete}] Given program $(h, o)$, a set of instances $I$ and a grammar $G$ for available comparisons, if there exists a keyword preorder $R$ satisfying (1) $\forall i \in I, CE(R, i) = \emptyset$, and (2) $R$ is constructed by $(\leq, o)$ and some comparisons in $G$, \mainname must terminate and return such a keyword preorder.
\end{theorem}
\begin{proof} Let $R$ be any solution satisfying the three conditions. According to Algorithm \ref{alg:preorder}, given a finite set of comparisons and a size limit, function \Search always terminate.

    We name an invocation of \Search good if the comparison space including all comparisons used in $R/\{(\leq, o)\}$ and $n_c$ is no smaller than $\textit{size}(R) - 1$. According to the iteration used to decide $C$ and $n_c$, for any $t$, there will be $t$ good invocations finished within finite time.
    
    Let $(op_1, k_1), \dots, (op_n, k_n)$ be an order of comparisons used in $R/\{(\leq, o)\}$ such that for any $x \in [1,n], (op_x, k_x)$ will be a valid comparison for function $\CC$ in the $x$th turn if $(op_1, k_1), \dots, 
    $$(op_{x-1}, k_{x-1})$ are selected in the previous terms. According to Lemma \ref{lemma:monoce}, such an order must exist.

    Suppose the error rate of \CC is at most $c$, i.e., the probability for \CC to exclude a valid comparison is at most $c$. For a good invocation of \Search, $R$ will be found if $\forall x \in [1, n], (op_x, k_x)$ is not falsely excluded in the $x$th turn by \CC. Therefore, the probability for $R$ to be found in a good invocation is at least $c' = (1 - c) ^n$, which is a constant.
    
    So, the probability for \mainname not to terminate after $t$ good invocations is at most $(1 - c')^t$. When $t \rightarrow + \infty$, this probability converges to $0$.
\end{proof}

\begin{lemma}[Lemma \ref{theorem:step2}] Given instance $i$ and program $\textit{prog}_1$ in Form \ref{form:thinning}, let $\textit{prog}_1'$ be result of Rule 1. If for any query $q$ and constructor $m$, Formula \ref{formula:2q} and Formula \ref{formula:2m} are satisfied respectively, $\textit{prog}_1 \sim_i \textit{prog}_1'$ holds.
    \begin{gather}
        \forall e \in RE(q, i), q\ e = \red{?q[q]}\ (\mathsf F[q]\red{?f_p}\ e)  \\
        \forall e \in RE(m, i), \red{?f_p}\ (m\ e) = \red{?c[m]}\ (\F[m]\red{?f_p}\ e)  
    \end{gather}
    \end{lemma}
    \begin{proof} Recall the form of $\textit{prog}_1$ and $\textit{prog}_1'$ as the following.
        \begin{align*}
            \textit{prog}_1 &= (rg((\textit{thin}\ \red{?R}) \circ \textit{cup} \circ \mathsf P\phi, \psi)_{\F}, o) \\
            \textit{prog}_1' &= (rg((\textit{thin}\ R') \circ \textit{cup} \circ \mathsf P{\phi}', \psi)_{\F}, \red{?q[o]})
        \end{align*}
    Comparing $\textit{prog}_1'$ with $\textit{prog}_1$, there are several expression-level differences: (1) all key functions in $\red{?R}$ are replaced with the corresponding $\red{?q}$, (2) the objective function is replaced with $\red{?q[o]}$, (3) all solution-related functions in $\phi$ are replaced with the corresponding $\red{?q}$ and $\red{?c}$. 
    
    Let $e_1, e_2$ be the small-step executions of $\textit{prog}_1$ and $\textit{prog}_1'$ on instance $i$, and let $e[k]$ be the $k$th program in execution $e$. Let us consider the following claim.
    \begin{itemize}
    \item For any $k$, $e_1[k]$ will be exactly the same as $e_2[k]$ after (1) replacing all solution-related functions with the corresponding $\red{?q}$ and $\red{?c}$, and (2) replacing all solutions with the outputs of $\red{?f_p}$. 
    \end{itemize}
    If this claim holds, the last programs in $e_1$ and $e_2$ must be the same because they are the outputs of $\textit{prog}_{1}$ and $\textit{prog}_{1}'$ and include neither functions nor solutions. So this claim implies $\textit{prog}_1 \sim_i \textit{prog}_1'$. 
    
    We prove this claim by induction on the number of steps. When $k = 0$, the claim directly holds because there is no solution constructed and the correspondence of functions is guaranteed by the construction of $\textit{prog}_1$. 
    
    Then for any $k >0$, consider the $k$th evaluation rule applied to $e_1$ and $e_2$. By the inductive hypothesis, these two evaluation rules must be the same. 
    \begin{itemize}
    \item If this evaluation rule relates to partial solutions, it must be the evaluation of a solution-related function. By the inductive hypothesis, (1) the scalar values in both inputs are exactly the same, and (2) the partial solutions used in $e_2$ are equal to the outputs of $\red{?f_p}$ on the partial solutions used in $e_1$. Therefore, the examples used in the synthesis task of Step 2 ensures that the outputs are still corresponding. At this time, the examples used in the synthesis task ensure that the evaluation result is still corresponding.
    \item If this evaluation rule does not relate to partial solutions, by the inductive hypothesis, the evaluation in $e_1$ and $e_2$ must be exactly the same. 
    \end{itemize}
    
    Therefore, the induction holds, and thus the claim holds.
    \end{proof}

    \begin{lemma}[Lemma \ref{theorem:step3}] Given instance $i$ and program $(r, o)$, where $r$ is a recursive generator, $(r^{\red{?f_m}}, o) \sim_i (r, o)$ if for any two states $s_1, s_2 \in (S_r\ i)$, $r\ s_1 \neq r\ s_2 \rightarrow \red{?f_m}\ s_1 \neq \red{?f_m}\ s_2$.
    \end{lemma}
    \begin{proof} Consider the following claim.
        \begin{itemize}
        \item Each time when $r^{\red{?f_m}}\ s$ returns, (1) the results is equal to $r\ s$, and (2) for any state $s' \in S_r\ i$, there is result recorded with keyword $\red{?f_m}\ s'$ implies that the results is $r\ s'$. 
        \end{itemize}
        If the claim holds, the lemma is obtained by $r^{\red{?f_m}}\ i = r\ i$.
        
        Let $r^{\red{?f_m}}\ s_1, \dots, r^{\red{?f_m}}\ s_n$ be all invocations of $r^{\red{?f_m}}$ during $r^{\red{?f_m}}\ i$ and suppose they are ordered according to the returning time. We prove the claim by induction on the prefixes of this sequence. For the empty prefix, the claim holds as the memoization space is empty. 
        
        Now, consider the $k$th invocation $r^{\red{?f_m}}\ s_k$. There are two cases. In the first case, there has been a corresponding result recorded in the memoization space. At this time, by the inductive hypothesis, this result must be equal to $r\ s_k$, and thus the claims still hold when $r^{\red{?f_m}}$ returns on $s_k$. 
        
        In the second case, there has not been a corresponding result recorded. By the inductive hypothesis, the results of the recursions made by $r^{\red{?f_m}}\ s_k$ must be the results of the corresponding recursions made by $r\ s_k$. Therefore, the execution of $r^{\red{?f_m}}\ s_k$ must be exactly the same with $r\ s_k$ and thus $r^{\red{?f_m}}\ s_k = r\ s_k$. By the examples used to synthesize $\red{?f_m}$, we know that for any other state $s \in S_r\ i$, $\red{?f_m}\ s = \red{?f_m}\ s_k$ implies that $r\ s = r\ s_k$, i.e., the memoized result in $r^{\red{?f_m}}\ s_k$.
        
        Therefore, the induction holds, and thus the claim holds.
        
        \end{proof}

        \begin{theorem} [Theorem \ref{theorem:correct}] Given input program $(h, o)$ where $h$ is a relational hylomorphism and a set of instances $I$, let $p^*$ be the program generated by \plusname with $I$. Then $\forall i \in I$, $(h, o) \sim_i p^*$.
        \end{theorem}
        \begin{proof} Because the correctness of Step 4 can be proved in the same way as Step 2, this theorem is directly from Theorem \ref{theorem:thinning-theorem}, Lemma \ref{theorem:step2}, Lemma \ref{theorem:step3}, and the correctness of Step 4.
        \end{proof}

        \begin{theorem} [Theorem \ref{theorem:efficiency}] Given input $(\hylo{\phi, \psi}_{\F}, o)$ and grammar $G$ specifying the program space for synthesis tasks, the program generated by \plusname must be pseudo-polynomial time if the following conditions are satisfied: (1) $\phi$, $\psi$ and programs in $G$ runs in pseudo-polynomial time, (2) each value and the size of each recursive data structure generated by the input program are pseudo-polynomial, (3) all operators in $G$ are linear, i.e., their outputs are bounded by a linear expression with respect to the input. 
        \end{theorem}
        \begin{proof}
            The time complexity of the resulting program can be decomposed into four factors: (1) the number of recursive invocations on the generator, (2) the maximum number of partial solutions returned by each invocation, (3) the time complexity of each invocation on the generator, and (4) the time complexity of each invocation on the scorer. To prove this theorem, we only need to prove that all of these four factors are pseudo-polynomial time.
        
            First, we prove that for any program in $G$ that returns a scalar value, its range is always pseudo-polynomial. For any such program $p$ in $G$, let $f_p(n, w)$ be a polynomial representing that the time cost of $p$ is at most $f_p(n, w)$ when $n$ scalar values in range $[-w, w]$ are provided as the input. 
        
            By the third precondition, there exists a constant $c$ such that for each operator $\oplus$ in $G$, for any input $\overline{x}$ and any output value $y \in \oplus \overline{x}$, $|y|$ is always at most $c \sum_{x \in \overline{x}} |x|$. 
        
            Suppose the size of program $p$ is $s_p$, which is a constant while analyzing the complexity of $p$. Now, suppose $n$ scalar values in range $[-w, w]$ are provided as the input to $p$. After executing the first operator, the sum of all available values is at most $f_p(n, w) \times cnw$, because there are at most $f_p(n, w)$ values due to the time limit and each value is at most $cnw$ according to the third precondition. Then, after the second operator, this sum increases to $f_p(n, w) \times c(f_p(n,w) \times cnw) = c^2 f_p(n,w)^2 \times nw$. In this way, we know that after executing all $s_p$ operators, the sum of all available values is at most $c^{s_p}f_p(n,w)^{s_p} \times nw$. Because $s_p$ is a constant, this upper bound is still pseudo-polynomial with respect to the input. 
        
            Second, we prove that the first two factors are pseudo-polynomial. The first factor is bounded by the range of $\red{?f_m}$, which is equal to the product of the ranges of key functions in $\red{?f_m}$. The second factor is bounded by the number of partial solutions returned by $\textit{thin}[\red{?R}]$. By Theorem \ref{theorem:thinning}, this value is also bounded by the product of the ranges of key functions in $\red{?R}$. Because the number of key functions in $\red{?f_m}$ and $\red{?R}$ are constants, we only need to prove that the range of each key function is pseudo-polynomial.
            \begin{itemize}
            \item For key functions in $\red{?f_m}$, by the second precondition, in the input program, both the size of a state and values in a state are pseudo-polynomial with respect to the global input. By our first result, we obtain that the range of each key function in $\red{?f_m}$ is pseudo-polynomial.
            \item For key functions in $\red{?R}$, by the second precondition, in the input program, both the size of a partial solution and values in a partial solution are pseudo-polynomial. By our first result, the scale of the new partial solution, i.e., the output of $\red{?f_p}$, must also be pseudo-polynomial. By the first result again, we obtain that the range of each key function in $\red{?R}$ is pseudo-polynomial.
            \end{itemize}
        
            Third, we prove that the third factor is pseudo-polynomial. According to Section \ref{section:plus}, the generator in the resulting program must be in the following form:
            $$
            rg(\textit{thin}[\red{?R}] \circ \textit{cup} \circ \mathsf P\phi', \psi')
            $$
            Therefore, the time complexity of each invocation can be further decomposed into four factors: (3.1) the time cost of $\textit{thin}[\red{?R}]$, (3.2) the time cost of $\phi'$, (3.3) the time cost of $\psi'$, and (3.4) the number of invocations of $\phi'$.
            \begin{itemize}
                \item According to Theorem \ref{theorem:thinning}, Factor 3.1 is bounded by the ranges of the key functions in $\red{?R}$, which has been proven to be pseudo-polynomial.
                \item For Factor 3.1 (3.2), the time cost of $\phi'$ ($\psi'$) is bounded by the time cost of $\phi$ ($\psi$) and all inserted program fragments $\red{?q}$ and $\red{?c}$ in Step 2 (Step 4). By the first precondition, their time costs are all pseudo-polynomial with respect to the new state, which has also been proven to be pseudo-polynomial in both values and scale. Therefore, the time cost of $\phi'$ ($\psi'$) is pseudo-polynomial.
                \item For Factor 3.3, by the first condition, the number of transitions (denoted as $n_t$) is pseudo-polynomial. The number of partial solutions returned by each recursive invocation (denoted as $n_p$) has been proven to be pseudo-polynomial, and the number of states (denoted by $n_s$) involved by a single transition is a constant. Therefore, the number of invocations of $\phi'$, which is bounded by $n_t \times n_p^{n_s}$, is also pseudo-polynomial.
            \end{itemize}
            Therefore, we prove that the third factor is also pseudo-polynomial with respect to the global input.
        
            At last, the fourth operator is also pseudo-polynomial because (1) the number of solutions and the scale of solutions are both pseudo-polynomial, and (2) the time complexity of the objective function, which is a program in $G$, is pseudo-polynomial by the first precondition. 
        
            In summary, all four factors are pseudo-polynomial, and thus we prove the target theorem.
           \end{proof}
\end{document}